%% file: main.tex
\documentclass[a4paper,USenglish,cleveref, autoref, thm-restate]{lipics-v2021}

\usepackage{graphicx} 
\usepackage{algorithm}
\usepackage{algpseudocodex}

\newtheorem{question}{Question}

\title{Distributed Triangle and Simplex Enumeration in Hypergraphs} 

\author{Duncan Adamson}{School of Computer Science, University of St Andrews, UK}{duncan.adamson@st-andrews.ac.uk}{0000-0003-3343-2435}{}

\author{Will Rosenbaum}{School of Computer Science and Informatics, University of Liverpool, UK}{w.rosenbaum@liverpool.ac.uk}{0000-0002-7723-9090}{}

\author{Paul G. Spirakis}{School of Computer Science and Informatics, University of Liverpool, UK}{P.Spirakis@liverpool.ac.uk}{0000-0001-5396-3749}{}

\authorrunning{Duncan Adamson, Will Rosenbaum, and Paul G. Spirakis} 

\Copyright{Duncan Adamson, Will Rosenbaum, and Paul G. Spirakis} 

\category{} 



\relatedversion{} 

\nolinenumbers 

\newcommand{\dft}[1]{\emph{\textbf{#1}}}
\newcommand{\paren}[1]{\left(#1\right)}

\newcommand{\set}[1]{\left\{#1\right\}}
\newcommand{\sqb}[1]{\left[#1\right]}
\newcommand{\mset}[1]{\left\{\!\middle\{#1\middle\}\!\right\}}
\newcommand{\sucht}{\,\middle|\,}
\newcommand{\rank}{\mathrm{rank}}

\newcommand{\KT}{\mathrm{KT}}

\newcommand{\abs}[1]{\left|#1\right|}

\DeclareMathOperator{\dist}{dist}
\DeclareMathOperator{\tr}{tr}
\DeclareMathOperator{\poly}{poly}

\newcommand{\tmix}{\tau_{\mathrm{mix}}}

\newcommand{\E}{\mathbb{E}}

\newcommand{\eps}{\varepsilon}

\newcommand{\calT}{\mathcal{T}}
\newcommand{\calE}{\mathcal{E}}
\newcommand{\calA}{\mathcal{A}}
\newcommand{\calS}{\mathcal{S}}

\DeclareMathOperator{\supp}{supp}

\newcommand{\bfe}{\mathbf{e}}
\newcommand{\bfw}{\mathbf{w}}
\newcommand{\bfX}{\mathbf{X}}
\newcommand{\bfY}{\mathbf{Y}}
\newcommand{\bfZ}{\mathbf{Z}}

\newcommand{\peel}{\textnormal{\textsc{Peel}}}
\newcommand{\ppeel}{\textnormal{\textsc{ParallelPeel}}}
\newcommand{\flood}{\textnormal{\textsc{Flood}}}
\newcommand{\enumeratetriangles}{\textnormal{\textsc{EnumerateMotifs}}}

\algblockdefx{ParFor}{EndParFor}[1]{\textbf{for} #1 \textbf{in parallel do}}{\textbf{end for}}

\begin{document}

\maketitle

\begin{abstract}
	In the last decade, subgraph detection and enumeration have emerged as central problems in distributed graph algorithms. This is largely due to the problems' theoretical challenges and practical applications. In this paper, we initiate the systematic study of distributed sub-hypergraph enumeration in hypergraphs. To this end, we (1)~introduce several computational models for hypergraphs that generalize the CONGEST model for graphs and evaluate their relative computational power, (2)~devise algorithms for distributed triangle and simplex enumeration in our computational models and prove their optimality in two such models by showing  matching lower bounds, (3)~introduce classes of sparse and ``everywhere sparse'' hypergraphs and describe efficient distributed algorithms for triangle and simplex enumeration in these classes, and (4)~describe general techniques that we believe to be useful for designing efficient algorithms in our hypergraph models.
\end{abstract}

\input{intro.tex}
\input{technical-overview.tex}
\input{preliminaries.tex}
\input{clique.tex}

\input{clique-ub.tex}

\input{clique-lb.tex}
\input{general-hypergraphs.tex}
\input{sparse-hypergraphs.tex}
\input{conclusion.tex}

\bibliography{bib}

\end{document}

%% file: intro.tex
\section{Introduction}
\label{sec:intro}


In distributed graph algorithms, a network is represented by a graph $G = (V, E)$. The set $V$ of vertices represents computing entities (nodes) and edges $e \in E$ represent communication links between vertices. Initially, vertices have little or no global information about the network $G$. Yet the goal is to find an efficient procedure for vertices to exchange messages and perform local computations such that the respective local outputs of the vertices satisfy some global specification. Some well-studied problems in this area include local symmetry breaking (e.g., finding maximal independent sets, maximal matchings, and proper colorings), global optimization problems (e.g., minimum spanning trees), identifying local structures (e.g. triangle/clique detection and enumeration), and approximate global optimization (e.g., maximum matchings and stable matchings). 


Fundamental to the study of these networks is the \emph{communication model}, the restrictions under which information may be shared between nodes. Our focus is on synchronous models, where executions proceed in synchronous rounds. In each round, vertices may perform unbounded local computation, send messages to neighboring vertices, and receive messages sent from neighboring nodes. The main measure of complexity is the number of rounds necessary and sufficient to perform a given task.
Three well-studied communication models introduced by~\cite{Linial1992Locality,Lotker2005MinimumWeight,Peleg2000Distributed} are:
\begin{description}
	\item[LOCAL]~\cite{Linial1992Locality} no restriction on message sizes, but communication is only between adjacent vertices
	\item[CONGEST]~\cite{Peleg2000Distributed} adjacent pairs of vertices can exchange messages of length $B$ bits (typically $B = O(\log n)$, where $n = \abs{V}$ is the number of vertices in the graph)
	\item[CLIQUE]~\cite{Lotker2005MinimumWeight} (or CONGESTED CLIQUE) and the related massively parallel computation (MPC) model~\cite{Karloff2010Model}, where each vertex can exchange short ($B$-bit) messages with all other vertices in the network.
\end{description}
These models each emphasize one or more ``dimensions'' of computational complexity for distributed systems, while abstracting away details such as local computational cost. Complexity in the LOCAL model measures how far in the local network each vertex must see in order to produce its local output. In contrast, complexity in the CLIQUE model measures how many bandwidth-restricted communication rounds are required to perform the task, independent of locality in the underlying graph $G$. Finally, complexity in the CONGEST model gives hybrid measure of complexity with restricted communication and locality in $G$. Thus, for a given task, understanding its computational complexity in these different models provides contrasting views on the aspects of distributed computing that make the problem easy or hard to solve.


The computational models above, however, abstract away many features of distributed systems that may be relevant to understanding the computational power of real-world systems. Two such shortcomings of these models is that they are homogenous (all vertices and edges have precisely the same computational power) and they do not model ``higher order'' interactions that involve more than two vertices. Indeed, many large-scale networks are hybrid or hierarchical in nature. The internet is best viewed as a ``network of networks'' with low-latency, high bandwidth connections between devices on a local area network, and potentially costlier connectivity between devices in different local networks~\cite{Atzori2010Internet,Peterson2011Computer}; data centers consist of thousands of servers, each with dozens or hundreds of cores on a single device~\cite{Benson2010Network}; social networks consist of many individuals that share different communication channels for different social groups to which each individual might belong~\cite{Kivela2014Multilayer}. In all of these settings, both the structure and the computational capabilities of the system are, arguably, better modelled as a hypergraph rather than a traditional graph~\cite{Battiston2020Networksa,Benson2016Higherorder,Bretto2013Hypergraph}.

A \dft{hypergraph} $H = (V, E)$ consists of a set, $V$, of vertices and a set, $E$, of edges that model connections between vertices. Unlike traditional graphs, each edge $e \in E$ may contain more than two vertices, thereby modelling higher order interactions between vertices. Hypergraphs offer greater descriptive power in defining and/or analyzing network structure compared to graphs. For example, specifying that a set of 5 devices share the same local network is more informative than saying that the 5 devices can communicate pair-wise with each other. Similarly, utilizing a shared email list or group chat is potentially both more efficient for planning social events, and the structure of such groups is more informative about the structure of a social network than pair-wise connectivity alone.

Given the generality and flexibility of hypergraphs in modelling networks, they have become an increasingly popular model for complex networks in data mining and machine learning~\cite{Antelmi2024Survey,Huang2015Scalable,Lee2025Survey}. Yet most of the work in this area focuses on centralized and parallel algorithms applied to hypergraphs, rather than modelling the underlying computational system as a hypergraph.

We believe that viewing underlying computational network as a hypergraph is illuminating in many cases. Indeed, such a view is natural for hybrid networks as described above, where the inclusion of hyperedges is used to model faster or more efficient communication between tightly-connected local networks, and the large-scale network structure is defined by the hypergraph. Our main motivation in this work is to initiate the systematic investigation of distributed hypergraph algorithms with congestion by analogy with the CLIQUE and CONGEST models for graphs. We focus attention on two generalizations of the triangle enumeration problem to hypergraphs: 3-cycle enumeration (hereafter referred to as hypergraph triangle enumeration) and simplex enumeration. Our reasons for considering these problem are twofold. On the theoretical side, triangle enumeration is in some sense the ultimate problem for understanding congestion in a network, as the problem requires a polynomial number of rounds the the CLIQUE and CONGEST models~\cite{Censor-Hillel2025Distributed,Chang2021Nearoptimal,Dolev2012Tri,Izumi2017Triangle,Pandurangan2021Distributed} while it can trivially be solved in a single round of the LOCAL model. On the applied side, the distribution of triangles (in both graphs and hypergraphs) and simplices are well-studied parameters that can reveal significant structural information about the network~\cite{Benson2018Simplicial,Jha2015SpaceEfficient,KoohiEsfahani2022LOTUS,LaRock2023Encapsulation,Sotiropoulos2021Triangleaware,Tsourakakis2011Counting,Chakrabarti2022Counting}.

\subsection{Our Contributions}\label{sec:contributions}

Our main contributions in this paper are to build a formal and conceptual framework for distributed algorithms in hypergraphs under congestion, and to analyze the complexity of the distributed triangle and simplex enumeration tasks within this framework. The primary contributions and results are as follows.

In Section~\ref{sec:computational-models}, we introduce six different models of for synchronous distributed computation that are all generalizations of the CONGEST model for graphs. Very briefly, these models are
\begin{itemize}
	\item PRIMAL CONGEST (PC) which is equivalent to the CONGEST model applied to the ``primal graph'' (i.e., the graph obtained by connecting each pair of vertices in each hyperedge with an edge)
	\item EDGE CLIQUE (EC) in which each hyperedge can simulate a single round of the CLIQUE model in each round
	\item EDGE BROADCAST (EB) in which each vertex in each hyperedge can broadcast a single message to all other vertices in the hyperedge
	\item EDGE UNICAST (EU) in which each vertex can send a single message to one neighboring vertex in each  incident edge in each round
	\item EDGE SOLOCAST (ES) in which each hyperedge specifies a single vertex to broadcast a message to an edge each round, and
	\item EDGE PAIRCAST (EP) in which each hyperedge specifies a single pair of incident vertices that can exchange messages each round.
\end{itemize}
We provide simulation results that characterize the relative power of these models (Propositions~\ref{prop:model-relationships} and~\ref{prop:pc-relationship}).

In Section~\ref{sec:clique}, we analyze the complexity of triangle and simplex enumeration in hypergraphs in the CLIQUE, PC, and EC models. In the triangle enumeration problem, each vertex should output a set of triangles (i.e., simple 3-cycles in $H$) such that every triangle in $H$ is output by precisely one vertex. Similarly, the simplex enumeration task requires that each simplex in the hypergraph is output by exactly one vertex.Recall that the \dft{rank} of a hypergraph $H = (V, E)$ is the largest cardinality of any hyperedge $e \in E$. We establish the following results.

\begin{theorem}[Informal]\label{thm:clique-inf}
	Consider the triangle and simplex enumeration tasks in hypergraphs of rank $r$. Then:
	\begin{enumerate}
		\item Triangle enumeration can be performed in $O(n^{r-5/3}/\log n)$ rounds in the CLIQUE model.
		\item Triangle enumeration requires $\Omega(n^{r-5/3}/\log n)$ rounds in the CLIQUE and PC models, and $\Omega(n^{1/3}/\log n)$ in the EC model.
		\item Triangle enumeration can be performed in $\widetilde{O}(n^{r-5/3})$ rounds in expectation in the PC and EC models.
		\item Simplex enumeration can be performed in $\Theta(n^{r - 2 + 1/(r+1)})$ rounds in the CLIQUE model, and this bound is tight.
	\end{enumerate}
\end{theorem}
The upper bounds of Theorem~\ref{thm:clique-inf} follow from natural generalizations the algorithms of Dolev, Lenzen, and Peled~\cite{Dolev2012Tri} (for the CLIQUE model) and Chang et al.~\cite{Chang2021Nearoptimal} for triangle enumeration in the PC and EC models. Our lower bounds for both triangle and simplex enumeration generalize the argument of Izumi and Le Gall~\cite{Izumi2017Triangle}. The main technical contribution in this section is the derivation of a combinatorial bound on the minimum number of (hyper)edges required to support a fixed number of triangles in hypergraphs (Theorem~\ref{thm:triangle-edge-bound}), which we believe to be novel.

While the upper and lower bounds of Theorem~\ref{thm:clique-inf} are tight in the CLIQUE and PC models, the exponential dependence on $r$ shows that these results may be impractical for hypergraphs with large rank. Thus, in Section~\ref{sec:sparse} we turn our focus to families of sparse hypergraphs. We consider two notions of sparseness. The first is ``global sparseness'' in which the sum of degrees of the vertices is $O(n)$. The second notion is ``everywhere sparseness'' in which the density of every sub-hypergraph is bounded by a parameter $\mu$. This notion generalizes the (asymptotically equivalent) notions of bounded arboricity and bounded degeneracy in graphs, though we believe its definition for hypergraphs to be novel. For sparse and everywhere sparse hypergraphs, we prove the following result.
\begin{theorem}[Informal]\label{thm:sparse-inf}
	The following bounds hold for triangle and simplex enumeration in sparse hypergraphs.
	\begin{itemize}
		\item Triangle and simplex enumeration in sparse hypergraphs on $n$ vertices can be performed in $O(n)$ rounds of the PC and EB models, and for every $\eps > 0$, $\Omega(n^{1-\eps})$ rounds are required for the PC model.
		\item In hypergraphs with maximum density $\mu$, triangle and simplex enumeration can be performed in $O(\mu r)$ rounds in the PC and EB models in hypergraphs with rank $r$, and the value of $\mu$ need not be known in advance. For every $\eps > 0$, $\Omega(\mu^{1-\eps})$ rounds are required for triangle and simplex enumeration in the PC model on hypergraphs with maximum density $\mu$ and constant rank $r$.
	\end{itemize}
\end{theorem}
In particular, we believe the maximum density parameter $\mu$ and associated algorithmic techniques in Section~\ref{sec:everywhere-sparse} to be of interest for other distributed computational problems on hypergraphs.

Finally, in Section~\ref{sec:conclusion}, we discuss other computational problems in the models introduced in Section~\ref{sec:computational-models}. We provide some generic simulation techniques that we believe will be helpful in generalizing CONGEST algorithms to hypergraphs, especially in the weaker hypergraph models, such as EDGE UNICAST (EU).

\subsection{Related Work}\label{sec:related-work}

The models we introduce in Section~\ref{sec:computational-models} are generalizations of the CONGEST model, formalized by Peleg~\cite{Peleg2000Distributed}. In addition to formalizing the CONGEST model, Peleg's text describes fundamental techniques and solutions to classic problems in distributed computing. The related CLIQUE (a.k.a.\ CONGESTED CLIQUE) model was introduced by Lotker et al.\ in~\cite{Lotker2005MinimumWeight}.

Dolev, Lenzen and Peled~\cite{Dolev2012Tri} devised an algorithm for triangle enumeration in graphs in the CLIQUE model in $O(n^{1/3}/\log n)$ rounds. Our CLIQUE algorithms for hypergraph triangle and simplex enumeration in Section~\ref{sec:clique-ub} (Theorems~\ref{thm:s-clique-ub} and~\ref{thm:clique-ub}) are generalizations of their algorithm. Pandurangan, Robinson, and Scquizzato~\cite{Pandurangan2021Distributed} proved a lower bound of $\Omega(n^{1/3} / \log^3 n)$ for (graph) triangle enumeration in the CLIQUE model, which was subsequently improved by Izumi and Le Gall~\cite{Izumi2017Triangle} to a tight lower bound of $\Omega(n^{1/3}/\log n)$. Our lower bounds for the CLIQUE and PC models generalize the argument of Izumi and Le Gall, which employs an information theoretic framework to bound the communication complexity of triangle enumeration. Izumi and Le Gall's argument employs a combinatorial lemma due to Rivin~\cite{Rivin2002Counting} that provides a lower bound on the number of edges in a graph containing a given number of triangles. This combinatorial lemma does not generalize to hypergraphs in the ``obvious'' way, but we employ Rivin's techniques to obtain a suitable generalization (Theorem~\ref{thm:triangle-edge-bound}) sufficient to derive our main CLIQUE lower bound (Theorem~\ref{thm:clique-lb}). Our simplex enumeration lower bound for the CLIQUE model follows a similar argument, instead relying upon a version of the Kruskal-Katona theorem (cf.~\cite{Jukna2011Extremal}) for a suitable bound on the minimum number of hyperedges required to contain a fixed number of simplices.

A nearly optimal algorithm for triangle enumeration in graphs in the CONGEST model was first described by Chang et al.~\cite{Chang2021Nearoptimal}. This result was achieved by combining two main components. The first component is to show that triangle enumeration can be performed in expected time $\widetilde{O}(n^{1/3} \tmix)$, where $\tmix$ is the mixing time of the graph. The main idea here is to simulate an analogue of Dolev et al.'s algorithm for the CLIQUE model using an expander routing technique developed by Ghaffari, Kuhn and Su~\cite{Ghaffari2017Distributed} and Ghaffari and Li~\cite{Ghaffari2018New}. The second component of Chang et al.'s algorithm is to develop an efficient distributed procedure for expander decomposition, whereby the graph is partitioned into components, each of which is an expander graph, by removing a small fraction of edges. Triangle enumeration can be performed efficiently on each component of the decomposition. The overall algorithm works by then recursively performing triangle enumeration on the remaining edges in the graph. Subsequent work has generalized and improved these results across several dimensions. We refer the reader to the survey of Censor-Hillel~\cite{Censor-Hillel2025Distributed} for recent developments on subgraph detection and enumeration. Our general upper bound for triangle enumeration in hypergraphs in the EC and PC models (Theorem~\ref{thm:ec-ub}) follows by showing that the techniques of Chang et al.\ generalize to our setting in a straightforward manner.

Our definition of the ``maximum density'' of a hypergraph (Definition~\ref{dfn:max-density}) is a hypergraph analogue of the characterization of graph arboricity given by Nash-Williams~\cite{Nash-Williams1964Decomposition}. For graphs, Barenboim and Elkin~\cite{Barenboim2008Sublogarithmic,Barenboim2013Distributed} showed that the vertices of graphs of bounded arboricity can be partitioned into $O(\log n)$ layers where each vertex has few neighbors in higher layers. They then use this result to derive efficient distributed algorithms for several fundamental problems in distributed computing. In Section~\ref{sec:everywhere-sparse} and show that an analogous decomposition result for hypergraphs can be applied to obtain an algorithm for hypergraph triangle and simplex enumeration parameterized by the hypergraph's maximum density (Theorem~\ref{thm:maximum-density-enumeration}).

The central role of triangles in the (graph) structure of social networks was first popularized by the work of Granovetter~\cite{Granovetter1973Strength}. The distribution of triangles and other ``motifs'' in networks has subsequently become foundational to network science~\cite{Milo2002Network,Newman2003Structure,Watts1998Collective}. Hypergraphs are becoming an increasingly popular tool for modeling complex networks. We refer the reader to the text of Bretto~\cite{Bretto2013Hypergraph} for general background and applications, as well as recent surveys by Lee et al.~\cite{Lee2025Survey} and Antelmi et al.~\cite{Antelmi2024Survey} for applications to data mining and machine learning (respectively). The distribution of triangles in hypergraphs modelling real-world networks has been studied, for example,  by Benson et al.~\cite{Benson2016Higherorder} and LaRock and Lambiotte~\cite{LaRock2023Encapsulation}.


In the other direction, distributed algorithms for hypergraphs have become an increasingly well-studied area, particularly concerning \emph{symmetry breaking} \cite{adamson2023distributed,adamson2025distributed,Balliu2023-distributed,Kutten2014Distributed,kuhn2018efficient}, a set of problems concerned with ensuring more even access to resources by splitting vertices either into colour classes \cite{adamson2023distributed,Kutten2014Distributed} or by assigning a local leader in the form of maximal independent sets \cite{adamson2025distributed,Balliu2023-distributed,kuhn2018efficient}. We note that in the above examples, communication is most commonly done in the LOCAL or CONGEST models on the primal graph. Indeed, it is open as to how the round complexity for undertaking such symmetry breaking under our more restrictive communication models, such as EDGE SOLOCAST or EDGE PAIRCAST.






%% file: technical-overview.tex
\subsection{Technical Overview}

Here, we provide a more detailed overview of our main results and techniques. Again, our main algorithmic problems are the triangle and simplex enumeration tasks. In a hypergraph $H = (V, E)$, we define a triangle to be a closed path of length $3$, i.e., a sequence of the form $T = (v_0, e_0, v_1, e_1, v_2, e_2)$ where $v_0, v_1, v_2 \in V$, $e_0, e_1, e_2 \in E$ and for $i = 0, 1, 2$, we have $v_{i} \in e_{i-1} \cap e_{i}$ where indices are computed modulo $3$. Such a triangle is \dft{simple} if $v_0, v_1, v_2$ and $e_0, e_1, e_2$ are pair-wise distinct, and \dft{induced} (a.k.a.~\dft{open}) if $v_i \notin e_{i+1}$ for $i = 1, 2, 3$. For the \dft{triangle enumeration task}, each vertex $v \in V$ must output a set $\calT_v$ of triangles such that the $\calT_v$ are pair-wise disjoint and $\calT = \bigcup_{v \in V} \calT_v$ contains all triangles in $H$. The algorithms we define can be adapted to output either all simple triangles or all induced (open) triangles in the same running time, so we do not specify if $\calT$ should be all simple triangles or all induced triangles in the task definition. Further, our lower bounds hold for both variants of the triangle enumeration task.

A hypergraph $H = (V, E)$ is \dft{$r$-uniform} if every edge $e \in E$ has cardinality $r$. In an $r$-uniform hypergraph, \dft{simplex} of dimension $r$ is a set $S = \set{v_0, v_1, \ldots, v_r}$ of $r+1$ vertices such that every subset of $S$ of cardinality $r$ is an edge in $E$. The \dft{simplex enumeration task} requires each vertex $v$ to output a set $\calS_v$ of simplices such that the $\calS_v$ are pair-wise disjoint and $\calS = \bigcup_{v \in V} \calS_v$ contains all simplices in $H$.

\subsubsection{Triangle and Simplex Enumeration in CLIQUE}

Our first main result (Theorems~\ref{thm:s-clique-ub} and~\ref{thm:clique-ub}, respectively) show that simplex and triangle enumeration in the CLIQUE model can be performed in $O(n^{r - 2 + 1/(r+1)})$ and $O(n^{r - 5/3}/\log n)$ rounds, respectively in $r$-uniform hypergraphs. Both of these results generalize a result of Dolev, Lenzen and Peled~\cite{Dolev2012Tri}, which corresponds to these two results when $r = 2$. 
In the case of simplex enumeration, we partition the set of $(r+1)$-tuples of distinct vertices in $V$, $V^{r+1}$, into $n$ parts, $T_1, T_2, \ldots, T_n$. Each vertex $v \in [n]$ is responsible for outputting all simplices contained in $T_v$. To this end, each $v$ must learn all the edges in the \emph{support} of $T_v$---the family of sets of vertices of cardinality $r$ whose elements appear in some tuple in $T_v$.

The partition is constructed in such a way that each $T_v$ is supported by $O(n^{r-1-1/(r+1)})$ edges, and each edge is in the support of $O(n^{1/(r+1)})$ parts $T_v$. Since each vertex has degree $O(n^{r-1})$, each vertex must send and receive at most $O(n^{r-1+1/(r+1)})$ edges in order to determine its output. Edges can be encoded with $O(r \log n)$ bits, so all messages can be routed in $O(r n^{r - 2 + 1/(r+1)})$ rounds deterministically using the routing algorithm of Lenzen~\cite{Lenzen2013Optimal}. We then improve this running time to $O(n^{r - 2 + 1/(r+1)} / \log n)$ by using bit-vectors to encode the (non)-existence of potential edges in $H$, thereby proving Theorem~\ref{thm:s-clique-ub}.

Theorem~\ref{thm:clique-ub} is proven similarly. Using the same partition $T_1, T_2, \ldots, T_n$ of \emph{triples} of vertices in $V$ from the simplex algorithm, vertex $v$ is responsible for outputting all triangles of the form $(v_0, e_0, v_1, e_1, v_2, e_2)$ where $(v_0, v_1, v_2) \in T_v$. To this end, it suffices for $v$ to learn all edges containing any two vertices from $\set{v_0, v_1, v_2}$. The partition guarantees that each vertex sends and receives $O(n^{r - 2/3})$ edges, so applying Lenzen routing and bit-vectors as before, we obtain a running time of $O(n^{r - 5/3}/\log n)$ rounds.

\subsubsection{Triangle Enumeration in PC and EC Models}

For triangle enumeration in the PC and EC models, we first observe that the EC model can trivially simulate the the PC model, hence it suffices to solve triangle enumeration in the PC model. In the PC model, our algorithm is an adaption of Chang et al.'s algorithm for triangle enumeration on graphs in the CONGEST model~\cite{Chang2021Nearoptimal}. The high level strategy of Chang et al.\ is to use a distributed expander decomposition to decompose the network into components, each of which has large conductance, by removing only a small fraction of edges. Then within each component, triangle enumeration is performed using a procedure that exploits the rapid mixing of high conductance graphs. Remaining triangles consisting only of edges between components are then dealt with recursively. For the rapidly mixing components of the expander decomposition, Chang et al.\ generalize the algorithmic techniques of Dolev et al.~\cite{Dolev2012Tri}. Specifically, they partition vertices according to their degree, and compute a derived partition of possible triangles (i.e., triples of vertices) such that each vertex is responsible for outputting all triangles consisting of triples of vertices assigned to it. This assignment is chosen such that each vertex must receive a number of edges proportional to its degree. Similarly, each vertex must send a number of messages (corresponding to its incident edges) that is proportional to its degree as well. The messages are then routed using an adaptation of the expander routing scheme of Ghaffari et al.~\cite{Ghaffari2017Distributed,Ghaffari2018New}. Under Chang et al.'s partition of triangles, each vertex $v$ is the source and destination of $O(n^{1/3}\deg(v))$ messages, and the message routing is performed in $\widetilde{O}(n^{1/3})$ rounds.

In order to adapt this procedure to hypergraph triangle enumeration in the PC model, we observe that the PC model is precisely the CONGEST model for the primal graph of $H  = (V, E)$---i.e., the graph $G = (V, E')$ where $\set{u, v} \in E'$ if and only if there exists an edge $e \in E$ such that $u, v \in E$. Thus, the different components of Chang et al.'s algorithm on $G$ can be simulated in $H$. In order to do this, we observe that hypergraph triangle enumeration in $H$ is equivalent to triangle enumeration in a multigraph $G_{\set{H}}$ that includes one copy of each edge $\set{u, v} \in E'$ for each (hyper)edge $e \in E$ that contains $u$ and $v$. Thus, the machinery of Chang et al.\ can be applied to $G$ with a simple adaptation: if $u$ sends an incident edge $\set{u, v}$ to a vertex $w$ in Chang et al.'s algorithm applied to $G$, $u$ instead sends all hyperedges $e$ containing both $u$ and $v$ to $w$ in our simulation. Since the number of such hyperedges is at most $\binom{n-2}{r-2} = O(n^{r-2})$, the message length and running time grows by a factor of $O(n^{r-2})$ compared to that of Chang et al.\ applied to $G$. Thus, the overall running time of our simulation is $\widetilde{O}(n^{1/3} \cdot n^{r - 2}) = \widetilde{O}(n^{r - 5/3})$.

\subsubsection{Lower Bounds}

In Section~\ref{sec:clique-lb}, we prove that simplex and triangle enumeration in the CLIQUE model (and hence the PC model as well) require $\Omega(n^{r-2+1/(r+1)})$ and  $\Omega(n^{r-5/3}/\log n)$ rounds, respectively thereby showing the algorithms above are optimal. Our argument generalizes the argument of Izumi and Le Gall~\cite{Izumi2017Triangle}, who prove the analogous lower bound for triangle enumeration in graphs. Specifically, we consider hypergraphs sampled from $G(n, r, 1/2)$, which is the distribution of $r$-uniform hypergraphs where each possible edge (i.e., subset of $r$ vertices) is included independently with probability $1/2$. We show that any algorithm that performs simplex or triangle enumeration for graphs sampled from this distribution with sufficiently high probability require $\Omega(n^{r-2+1/(r+1)})$ and $\Omega(n^{r - 5/3}/\log n)$ rounds (respectively) in expectation.

Our lower bound argument follows the information theoretic framework developed by Izumi and Le Gall. The basic idea is that in order for a vertex $v$ to output a set $\calS_v$ of simplices (respectively a set $\calT_v$ of triangles), the information content of the messages received by $v$ must be at least the information content of the (existence of) the edges contained in the simplices (triangles) in $\calS_v$ ($\calT_v$). This statement is formalized in Lemma~\ref{lem:triangle-information}, and its proof is identical to that of the analogous statement in~\cite{Izumi2017Triangle}. We can then obtain round lower bounds by observing that the the total message length received must be at least the expected entropy of the messages received. Since each vertex can receive only $O(n \log n)$ bits each round, an entropy lower bound of $h$ yields a round lower bound of $h / n \log n$.

For triangle enumeration, in order to prove a lower bound on the maximum entropy of messages received by some vertex $v$, we argue as follows. For each $v \in V$, let $\calT_v$ denote the set of triangles output by $v$, and take $\calT = \bigcup \calT_v$ be the set of all triangles in $H$ sampled from $G(n, r, 1/2)$. We first show that with large constant probability, $\abs{\calT} = \Omega(n^{3r -3})$ (for fixed constant $r$). Let $\bfw = \arg \max_v \abs{\calT_v}$ be the vertex that outputs the most triangles. By the pigeonhole principle, $\abs{\calT_{\bfw}} = \Omega(n^{3r-4})$ with constant probability.

Next, we show that the triangles in $\calT_\bfw$ must contain many edges. To this end, we establish the following bound (cf.~Theorem~\ref{thm:triangle-edge-bound}): if $H$ is an $r$-uniform hypergraph with $n$ vertices and $m$ edges that contains $t$ triangles, then
\begin{equation}\label{eqn:triangle-bound}
	t \leq \frac{\sqrt{2}}{3}\paren{m \binom{n-3}{r-2} \binom{r}{2}}^{3/2}.
\end{equation}
Taking $t = \Omega(n^{3r - 4})$ and solving for $m$ gives $m = \Omega{n^{r-2/3}}$ for constant $r$. Thus in outputting $\calT_\bfw$, $\bfw$ must witness $\Omega(n^{r-2/3})$ edges.

The bound~(\ref{eqn:triangle-bound}) generalizes a bound of Rivin~\cite{Rivin2002Counting}, which corresponds to~(\ref{eqn:triangle-bound}) in the case of graphs (i.e., $r = 2$). Our proof of~(\ref{eqn:triangle-bound}) essentially follows the argument of Rivin, where the crucial insight again is to associate triangles in a hypergraph $H$ with triangles in the multigraph $G_{\set{H}}$ introduced above. In $G_{\set{H}}$, the maximum multiplicity of an edge is $\mu \leq \binom{n-2}{r-2} = O(n^{r-2})$. By applying standard tools in algebraic graph theory, we show that for such a multigraph with $t$ triangles and $m$ edges, the number of triangles satisfies $(6t)^2 \leq (2 \mu m)^3$, whence~(\ref{eqn:triangle-bound}) follows.

In order to prove the lower bound for simplex enumeration, we again must argue that simplices $\calS_\bfw$ output by some vertex $\bfw$ must contain many edges in $H$ sampled from $G(n, r, 1/2)$. We show that $H$ contains $\Theta(n^{r+1})$ simplices with constant probability, hence some $\bfw$ must output $\Theta(n^r)$ simplices with constant probability. We then invoke a version of the Kruskal-Katona theorem (Theorem~\ref{thm:kruskal-katona}) in order to conclude that the simplices output by $\bfw$ must contain $\Omega(n^{r-1+1/(r+1)})$ edges.

Finally, we conclude our proof of the main simplex and triangle counting lower bounds by following the argument of Izumi and Le Gall. We take $\calE = (\bfe_1, \bfe_2, \ldots, \bfe_m)$ to be the indicator vector for the existence of each of the $m = \binom{n}{r}$ possible edges in $H$, and we let $\pi_\bfw$ denote the transcript of messages received by $\bfw$. By applying standard tools from information theory, along with Lemma~\ref{lem:triangle-information}, we argue that the entropy of $\bfw$ satisfies $H(\pi_\bfw) \geq \E[\abs{E(\calT_\bfw)}] - H(\rho_\bfw)$, where $E(\calT_\bfw)$ is the set of edges contained in the triangles output by $\bfw$ and $H(\rho_\bfw)$ is the entropy of the initial state of $\bfw$. The main lower bound then follows by applying the bounds above for $\abs{E(\calT_\bfw)}$ and $\abs{E(\Sigma_\bfw)}$ along with the observation that $H(\rho_\bfw) = O(n^{r-2})$.

\subsubsection{Sparse Hypergraphs}

We say that a hypergraph $H = (V, E)$ is \dft{sparse} if $\sum_{e \in E} \abs{e} = \sum_{v \in V} \deg(v) = O(n)$. We show that if $H$ has maximum degree $\Delta$, then simplex and triangle enumeration can be performed in $O(\Delta)$ rounds of the EDGE BROADCAST (EB) and PRIMAL CONGEST (PC) models (Theorem~\ref{thm:bounded-degree-ub}), and in particular, for sparse hypergraphs, this running time is $O(n)$. The algorithm is simple: we simply compute unique edge IDs, and each vertex broadcasts the IDs of its incident edges to its neighbors. It is straightforward to show that this is sufficient to solve simplex and triangle enumeration. To compute short ($O(\log n)$ bit) unique IDs, we take the ID of an edge $e$ to be the concatenation of the smallest vertex ID of a vertex $v$ in $e$ concatenated with $v$'s port number associated to the edge $e$.

We show that this $O(n)$ upper bound for sparse graphs is essentially tight in the PC model by reduction to the CLIQUE lower bound (Theorem~\ref{thm:sparse-lb}). Specifically, for any $\eps > 0$, we can find a rank $r$ and integer $n' \leq n$ such that the CLIQUE lower bound of Theorem~\ref{thm:clique-lb} for $H' \sim G(n', r, 1/2)$ is $\Omega(n^{1 - \eps})$, and $H'$ contains $O(n)$ edges (for constant $r$). We then form a distribution of hypergraphs $H$ on $n$ vertices that contain copies of $H'$ that are sparse and such that we can apply the CLIQUE lower bound in the PC model.

\subsubsection{Everywhere Sparse Hypergraphs}

For a hypergraph $H = (V, E)$, we define the \dft{density} of a subset $U \subseteq V$ of vertices to be
\begin{equation*}
	\mu_H(U) = \frac{1}{\abs{U}} \sum_{e \in E[U]} \abs{e},
\end{equation*}
where $E[U] = \set{e \cap U \sucht e \in E, \abs{e \cap U} \geq 2}$ is the set of induced edges $U$. We define the \dft{maximum density} of $H$ to be $\mu(H) = \max_{U} \mu_H(U)$. We show that maximum density has properties analogous to that of arboricity and degeneracy of graphs. Specifically, we show that if $H$ has maximum density $\alpha$, then $V$ can be partitioned into $\ell = O(\log n)$ layers $V = L_1 \cup L_2 \cup \cdots \cup L_\ell$ such that each vertex $v \in L_i$ has degree at most $2 \alpha$ in $H$'s restriction to $V_i = L_i \cup L_{i+1} \cup \cdots \cup L_\ell$. In Section~\ref{sec:everywhere-sparse}, we show that this structural result for hypergraphs with bounded maximum density can be applied to give efficient algorithms for simplex and triangle enumeration. The algorithmic techniques mirror those of Barenboim and Elkin's work~\cite{Barenboim2008Sublogarithmic,Barenboim2013Distributed} on arboricity bounded graphs.

We first define a distributed algorithm $\peel(\alpha)$ for the EB model that takes an estimate $\alpha$ of the maximum density as input. If the maximum density of $H$ is at most $\alpha$, then the algorithm produces a layered decomposition of the vertices as above, where each vertex $v \in V$ learns which layer $L_i$ in which it is contained. The algorithm follows a simple greedy approach. All vertices are initially active. In the first round, all vertices with degree at most $2 \alpha$ become inactive, join $L_1$, and broadcast this fact to their neighbors. In subsequent rounds, each vertex keeps track of its incident edges still containing active vertices. If in round $i$, $v$ is incident to at most $2 \alpha$ edges containing active vertices, then $v$ becomes inactive and joins $L_i$. Thus, if each vertex is given a correct upper bound $\alpha$ on the maximum density, the vertices can compute a layered decomposition in $O(\log n)$ rounds of the EB model. Conversely, we show that if some vertex $v$ fails to become inactive after $\log n$ rounds, then the maximum density of $H$ is at least $\alpha / 2$.

Given layered decomposition $L_1, L_2, \ldots, L_\ell$ with $\ell = O(\log n)$, let $V_i = L_i \cup L_{i+1} \cup \cdots \cup L_\ell$.  We can perform triangle and simplex enumeration as follows. In rounds $i = 1, 2, \ldots, \ell$, each vertex $v \in L_i$ broadcasts each nonempty $e \cap V_i$ for each of $v$'s incident edge $e$ to all of its neighbors. We show that after performing this broadcast, each vertex $v \in L_i$ has sufficient information to enumerate all of its incident triangles and simplices containing vertices in $V \setminus V_{i+1}$. Thus, after $O(\log n)$ phases, all triangles and simplices are enumerated. Each phase, each vertex sends at most $2 \alpha$ edges, each containing at most $r$ vertices, so the overall running time is $O(\alpha r)$.

The approach above assumes that the vertices are given some \emph{a priori} upper bound $\alpha$ on the maximum density. We next devise a procedure that computes local estimates of $\alpha$ that are sufficient to solve triangle and simplex enumeration in time $O(\mu(H) r + \log n)$ without the assumption that a faithful estimate $\alpha \geq \mu(H)$ is given, so long as some \emph{a priori} upper bound on the maximum density $M = n^{O(1)}$ is known to all vertices. The procedure works as follows. Each vertex runs the ``peeling'' procedure for computing a layered decomposition described above for geometrically increasing estimates $\alpha = 2, 4, 8,\ldots,M$. These can all be run in parallel using $O(\log M)$ bit messages. Each vertex $v$ then records the smallest estimate $\alpha_v$ for which $v$ terminated after at most $\log n$ rounds with $\alpha = \alpha_v$. We prove that $\max_v \alpha_v$ is a constant factor approximation of the true maximum density $\mu(H)$. Further, we show that if each vertex sets $\beta_v$ to be the maximum value of $\alpha_u$ for $u$ in $v$'s distance $\log n$ neighborhood, then the triangle enumeration procedure above can be implemented where each vertex $v$ uses its local estimate $\beta_v$ of the maximum density. In the end, we obtain an algorithm that solves triangle and simplex enumeration in $O(\mu(H)r + \log n)$ rounds in the EB and PC models.

Finally, we show that the dependence on $\mu(H)$ is tight for the PC model. Specifically, for every $\eps > 0$, there is a distribution of hypergraphs with maximum density at most $\mu$ which require $\Omega(\mu^{1 - \eps}/\log n)$ rounds in the PC model. Again this argument is via reduction to the CLIQUE lower bound construction. We show that for suitably chosen parameters, $n'$ and $r = O(1)$, $H \sim G(n', r, 1/2)$ has maximum density $\Theta(\mu^{1 - \eps})$ with high probability, and the PC lower bound follows from our general CLIQUE lower bound, Theorem~\ref{thm:clique-lb}.

%% file: preliminaries.tex
\section{Preliminaries}
\label{sec:preliminaries}

\subsection{Hypergraph Definitions}

A \dft{hypergraph} $H = (V, E)$ consists of a set $V$ of \dft{vertices} and a set $E$ of \dft{(hyper)edges}, where each $e \in E$ is a non-empty subset of $V$. The \dft{rank} of $H$ is $r(H) = \sup_{e \in E} \set{\abs{e}}$. We say that $H$ is \dft{$r$-uniform} if every edge $e \in E$ satisfies $\abs{e} = r$.

The \dft{degree} of a vertex $v$, denoted $\deg(v)$ is the number of edges containing $v$:
\begin{equation}
	\deg(v) = \abs{\set{e \in E \sucht v \in e}}.
\end{equation}
Throughout this paper, we will use $\Delta = \max_v \deg(v)$ to denote the maximum degree of a vertex in $H$. Similarly, for any pair of vertices $u, v \in V$, we define the \dft{pair degree}
\begin{equation}
	\deg(u, v) = \abs{\set{e \in E \sucht u, v \in e}}
\end{equation}
to be the number of edges containing both $u$ and $v$, and we denote $\Delta_2 = \max_{u, v} \deg(u, v)$. Recall that a hypergraph $H$ is \dft{linear} if for every pair of edges $e \neq e'$, we have $\abs{e \cap e'} \leq 1$. Thus, $H$ is linear if and only if $\Delta_2 \leq 1$.

The \dft{neighborhood} of $v$, denoted $N(v)$ is the set of vertices that share an edge with $v$:
\begin{equation}
	N(v) = \bigcup_{e\,:\,v \in e} e.
\end{equation}


Given a hypergraph $H = (V, E)$, a \dft{path} $P$ is a sequence $P = (v_0, e_0, v_1, e_1, v_2, e_2, \ldots, e_{\ell - 1}, v_\ell)$, where each $e_i \in E$ and $v_i \in e_{i-1}, e_i$. Here, $\ell$ is the \dft{length} of the path. We say that a path $P$ is\dots
\begin{itemize}
	\item[\dots] \dft{vertex simple} if for all $i \neq j$, $v_i \neq v_j$;
	\item[\dots] \dft{edge simple} if for all $i \neq j$, $e_i \neq e_j$;
	\item[\dots] \dft{simple} if it is both vertex simple and edge simple;
	\item[\dots] \dft{induced} if it is simple and for all $i$ and $j \neq i-1, i$ we have $v_i \notin e_j$.
\end{itemize}
A \dft{cycle} is a path $P = (v_0, e_0, \ldots, v_\ell)$ such that $v_\ell = v_0$. Cycles can be vertex simple, edge simple, simple, and induced, defined analogously as for paths above.

\begin{example}
	\label{ex:triangle}
	Consider a 4-uniform hypergraph $H$ where $V = \set{1, 2, \ldots, 9}$. Then $H$ contains a $3$-cycle with vertex set $T = \set{1, 2, 3}$ $v_0 = 1, v_1 = 2, v_2, = 3$ for the following edge sets:
	\begin{itemize}
		\item $e_0 = e_1 = e_2 = \set{1, 2, 3, 4}$
		\item $e_0 = \set{1, 2, 3, 4}$, $e_1 = \set{2, 3, 5, 6}$, $e_2 = \set{1, 3, 7, 8}$
		\item  $e_0 = \set{1, 2, 4, 5}$, $e_1 = \set{2, 3, 4, 6}$, $e_2 = \set{1, 3, 4, 7}$
		\item $e_0 = \set{1, 2, 4, 5}$, $e_1 = \set{2,3,6,7}$, $e_2={1,3,8,9}$
	\end{itemize}
	The first triangle is not (edge)-simple, while the second is simple, but not induced. The last two triangles are induced.
\end{example}

We use the following definitions of hypergraph triangles throughout the remainder of the paper.

\begin{definition}
	\label{dfn:triangle}
	A \dft{triangle} $T = (v_0, e_0, v_1, e_1, v_2, e_2, v_0)$ is a simple cycle of length three. $T$ is an \emph{induced triangle} if $T$ is an induced cycle.
\end{definition}

The results we provide for the remainder of the paper apply to both the tasks of enumerating (not necessarily induced) triangles and induced triangles. Thus, in formal statements we do not specify if all triangles should be enumerated or if only induced triangles should be enumerated.

\begin{definition}\label{dfn:simplex}
	If $H = (V, E)$ is an $r$-uniform hypergraph, then an \dft{$r$-dimensional simplex} is a set of $r+1$ vertices, $S = \set{v_0, v_1, \ldots, v_{r+1}}$, such that for every subset $S' \subseteq S$ of size $r$, we have $S' \in E$.
\end{definition}

We note that in the case of $2$-uniform hypergraphs, i.e., graphs, both definitions~\ref{dfn:triangle} and~\ref{dfn:simplex} reduce to the standard definition of triangle in a graph. Thus, both concepts can be seen as generalizations of triangles in hypergraphs.

We say that a hypergraph $H = (V, E)$ is \dft{connected} if for every pair of vertices $u, w \in V$ there is a path $P = (u, e_0, v_1, e_1, \ldots e_{\ell_1}, w)$ from $u$ to $w$. The \dft{distance} between vertices $u$ and $w$, denoted $\dist_H(u, v)$ is the length of the shortest such path between $u$ and $v$, and $\dist_H(u, v) = \infty$ if there is not such path. The \dft{diameter} of a hypergraph $D(H)$ is defined by $D(H) = \sup_{u, v \in V} \dist_H(u, v)$.

Given a hypergraph $H = (V, E)$ we define the \dft{primal graph} $G(H)$ of $H$ to be the graph $G = (V, E')$, where $\set{u, v} \in E'$ if and only if there exists an edge $e \in E$ with $u, v \in e$. That is, $G(H)$ is obtained by replacing each (hyper)edge $e \in E$ with a clique in $G$.

\subsection{Computational Models}
\label{sec:computational-models}

We view each hypergraph as a communication network by analogy with distributed graph algorithms. For a hypergraph $H = (V, E)$, we assume that each vertex $v \in V$ has a unique ID of $O(\log n)$ bits, where $n = \abs{V}$ is the number of vertices in the hypergraph. We assume that each vertex's ID is initially known to that vertex, and possibly its neighboring vertices as well.

We focus on synchronous message passing models that are analogous to the (congested) CLIQUE, LOCAL, and CONGEST models. Throughout our discussion of computational models, we use the parameter $B$ as the number of bits communicated per interaction per communication round. Traditionally (e.g., for the CONGEST model), we take $B = O(\log n)$, but for the sake of generality we treat $B$ as a parameter in defining our models. We assume that initially each vertex knows its own ID, degree, as well as a list of neighboring vertices in each incident edge. Further, we assume that each vertex initially holds a port numbering of its incident edges so that it can refer unambiguously to its first, second, etc., incident edges.

One approach to generalizing the LOCAL and CONGEST models to hypergraphs is simply to use the primal graph $G(H)$ as the underlying communication network. In each case, we assume that the input of each vertex contains an ordered list of its adjacent edges. We formalize the communication restrictions in these models as follows:
\begin{description}
	\item[CLIQUE] each round, each vertex $v \in V$ can send (distinct) $B$-bit messages to all other vertices $w \in V$. This is equivalent to the ``congested clique'' model for standard graphs.
	\item[PRIMAL CONGEST] (PC) each round, each vertex $v \in V$ can send (distinct) $B$-bit messages to each adjacent vertex $w \in N(v)$.
	\item[PRIMAL LOCAL]  (PL) each round, each vertex $v \in V$ can send (distinct) messages of unbounded size to each adjacent vertex $w \in N(v)$.
\end{description}
Note that the PC and PL models are equivalent to the CONGEST and LOCAL models on the primal graph $G(H)$.

While the PC model may be a ``natural'' generalization of the CONGEST model for hypergraphs in some contexts, it ignores the structure of (hyper)edges in the network, and reduces all communication to pair-wise adjacency. Depending on what the edges in the hypergraph are meant to model, this may afford either too much or too little computational power to the model. In order scenarios in which each (hyper)edge represents a shared communication channel, we define five generalizations of the CONGEST model to hypergraphs. The power of these models all differ as a function of the rank $r$ of the hypergraph, though for fixed constant rank they are equivalent to the (PRIMAL) CONGEST model.


We define our distributed hypergraph models as follows:

\begin{description}
	\item[EDGE CLIQUE] (EC) each round, for each vertex $v \in V$ and each incident edge $e \ni v$, $v$ can send (distinct) $B$ bit messages to each other vertex $w \in e$. This is equivalent to each (hyper)edge acting as in the congested clique model for standard graphs.
	\item[EDGE BROADCAST] (EB) each round, for each vertex $v$ and incident edge $e \ni v$, $v$ can broadcast a single $B$-bit message to all vertices $w \in e$. Each $v$ may broadcast different messages to different incident edges. This is equivalent to each (hyper)edge acting as in the ``broadcast congested clique'' model for standard graphs.
	\item[EDGE UNICAST] (EU) each round, for each vertex $v$ and incident edge $e \ni v$, $v$ may choose a single incident vertex $w \in e$ and send a single $B$-bit message to $w$. Thus, if $v$ has degree $d(v)$, $v$ sends (at most) $d(v)$ distinct messages to (at most) $d(v)$ neighboring vertices.
	\item[EDGE SOLOCAST] (ES) each round $t$, each edge $e$ specifies a unique vertex $v \in e$ that broadcasts a single $B$ bit message that is received by all other vertices $w \in e$ in the round.
	\item[EDGE PAIRCAST] (EP) each round $t$, each edge $e$ specifies a unique pair of incident vertices $u, v \in e$ that can exchange $B$ bit messages in the round.
\end{description}

EC, EB, EU, ES, and EP can each be viewed as a generalization of the standard CONGEST model for graphs, as each eas equivalent to CONGEST for $2$-uniform hypergraphs (i.e., graphs). We note that EC is at least as strong as PC, but the power of the remaining models' computational power are generally incomparable to PC due to the possible existence of parallel edges between pairs of vertices.

Following Awerbuch et al.~\cite{Awerbuch1990Tradeoff}, we can quantify the initial information known to each vertex $v \in V$ as follows. We assume in all cases that $v$ knows $\deg(v)$, i.e., its number of incident edges. For each non-negative integer $k$, let $N_k(v) = \set{u \sucht \dist_H(u, v) \leq k}$ be the set of vertices within distance $k$ of $v$. Similarly, we define the set of edges $E_k(v) = \set{e \in E \sucht \exists u \in e,\, \dist_H(v, u) \leq k - 1}$. The \dft{$k$-neighborhood} of $v$ contains all vertices in $N_k(v)$ and all edges in $E_k(v)$. In the \dft{initial knowledge} model $\KT_k$, each vertex $v \in V$ initially knows it $k$-neighborhood. Note that in the case of $\KT_0$, $v$ does not initially know the vertices in $N(v)$, hence it must refer to edges only by some (fixed) port numbering.

For the remainder of the paper, we assume initial knowledge $\KT_1$. Indeed, $\KT_0$ seems unnecessarily restrictive for many applications, while triangle listing is trivial in $\KT_2$. We further assume that $B = \Theta(\log n)$ bit messages are used for each interaction.

The following proposition quantifies relationships between the computational power of the five generalizations of CONGEST defined above. 

\begin{proposition}\label{prop:model-relationships}
	Suppose $H = (V, E)$ is a hypergraph of rank $r$. Then under initial knowledge $\KT_1$ we have the following reductions:
	\begin{itemize}
		\item A single round in the EDGE CLIQUE model can be simulated by $O(r)$ rounds in the EDGE BROADCAST and EDGE UNICAST models and $O(r^2)$ rounds of the EDGE SOLOCAST and EDGE PAIRCAST models.
		\item A single round in the EDGE BROADCAST model can be simulated in a single round of EDGE CLIQUE, in $O(r)$ rounds of EDGE UNICAST and EDGE SOLOCAST, and in $O(r^2)$ rounds of EDGE PAIRCAST.
		\item A single round in the EDGE UNICAST model can be simulated in $O(1)$ rounds in the EDGE CLIQUE and EDGE BROADCAST models and $O(r)$ rounds of EDGE SOLOCAST and EDGE PAIRCAST.
		\item A single round of EDGE SOLOCAST can be simulated by a single round of EDGE CLIQUE and EDGE BROADCAST, $O(\log r)$ rounds of EDGE UNICAST, and $O(r)$ rounds of EDGE PAIRCAST.
		\item A single round of EDGE PAIRCAST can be simulated by a single round of all the other models.
	\end{itemize}
	Therefore, if a given task has complexities $T_{EC}$, $T_{EB}$, $T_{EU}$, $T_{ES}$, and $T_{EP}$ in the EC, EB, EU, ES and EP models, respectively, then we have:
	\begin{align}
		T_{EC} & = O(T_{EU}), O(T_{EB}), O(T_{ES}), O(T_{EP})                                      \\
		T_{EB} & = O(r \cdot T_{EC}), O(T_{EU}), O(T_{ES}), O(T_{EP})                              \\
		T_{EU} & = O(r \cdot T_{EC}), O(r \cdot T_{EB}), O(\log r \cdot T_{ES}), O(r \cdot T_{EP}) \\
		T_{ES} & = O(r^2 \cdot T_{EC}), O(r \cdot T_{EB}), O(r \cdot T_{EU})                       \\
		T_{EP} & = O(r^2 \cdot T_{EC}), O(r^2 \cdot T_{EB}), O(r \cdot T_{EU}), O(r \cdot T_{ES})
	\end{align}
	These relationships are depicted in Figure~\ref{fig:model-relationships}.
\end{proposition}
\begin{proof}
	All of the simulations above are trivial, except for simulating a round of ES with $O(\log r)$ rounds of EU. The basic idea is as follows. Suppose a vertex $v$ wishes to broadcast a message $M$ to its incident edge $e = \set{v_0, v_1, \ldots, v_{r-1}}$. Without loss of generality, assume that $v = v_0$ and that the IDs of $v_0, v_1, \ldots, v_{r-1}$ are strictly increasing. Thus, under the $\KT_1$ assumption, each vertex $w \in e$ knows the index $i$ of each vertex in $v_i \in e$. To broadcast $M$, $v_0$ initially sends $M$ to $v_1$. In each subsequent round $t = 1, 2, \ldots, \log r$, each vertex $v_i$ with $i < 2^t$ send $M$ to vertex $v_{i + 2^{t-1}}$. The correctness of this procedure follows from a straightforward induction argument.
\end{proof}

\begin{figure}
	\begin{center}
		\begin{tikzpicture}
			\node (EC) at (2, 3) [draw, circle] {EC};
			\node (EB) at (2, 1.5) [draw, circle] {EB};
			\node (ES) at (0, 0) [draw, circle] {ES};
			\node (EU) at (2, 0) [draw, circle] {EU};
			\node (EP) at (4, 0) [draw, circle] {EP};
			\draw[->, thick] (EB) -- node[right] {$r$} (EC);
			\draw[->, thick] (EU) -- node[right] {$r$} (EB);
			\draw[<->, thick] (EP) -- node[below] {$r$} (EU);
			\draw[->, thick, bend left=30] (ES) to node[above] {$r$} (EU);
			\draw[->, thick, bend left=30] (EU) to node[above] {$\log r$} (ES);
			\draw[<->, thick, bend left=50] (EP) to node[below] {$r$} (ES);
			\draw[->, thick, bend left=30] (ES) to node[left] {$r$} (EB);
			\draw[->, thick, bend right=30] (EP) to node[right] {$r$} (EB);
		\end{tikzpicture}
	\end{center}
	\caption{A figure demonstrating the relationships between models given in Proposition~\ref{prop:model-relationships}. Directed edges correspond to simulation results, where the edge label gives the multiplicative overhead of the simulation. The vertical position of each model in the figure indicates the power of the model, with EC and EB being able to simulate all lower models without computational overhead.\label{fig:model-relationships}}
\end{figure}
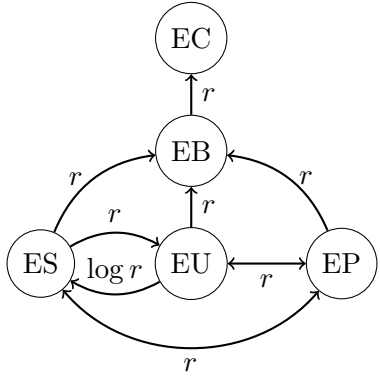

According to Proposition~\ref{prop:model-relationships}, the five models EC, EB, EU, ES and EP are equivalent up to a factor of $O(r^2)$. In particular, in hypergraphs with bounded rank (i.e., $r = O(1)$) they are asymptotically equivalent.

As noted above, the computational power of EC, EB, EU, ES, and EP relative to CLIQUE and PC depends on parameters of the hypergraph. For example, if vertices $u, v \in V$ are contained in $k$ distinct (hyper)edges then the former models allow for $u$ and $v$ to potentially exchange $\Theta(k B)$ bits in a single round: $B$ bits per shared edge. The maximum 2-degree $\Delta_2$ of $H$ gives an upper bound on the number of parallel messages that can be sent between any adjacent pair of vertices in $H$. Thus, we can formalize the relationship between, for example, PC and EC, as follows.

\begin{proposition}\label{prop:pc-relationship}
	Let $H = (V, E)$ be a hypergraph with maximum $2$-degree $\Delta_2$. Then a single round of computation in the PC model can simulated in a single round of EC, and a single round of computation in the EC model can be simulated in $\Delta_2$ rounds of the PC model. Thus, if a task has computational complexity $T_{PC}$ in the PC model and $T_{EC}$ in the EC model, then
	\begin{equation}
		T_{EC} \leq T_{PC} \leq \Delta_2 T_{EC}.
	\end{equation}
\end{proposition}

By combining the conclusions of Propositions~\ref{prop:model-relationships} and~\ref{prop:pc-relationship}, we can obtain analogous relationships between PC and the other models listed in Proposition~\ref{prop:model-relationships}. Combining Proposition~\ref{prop:pc-relationship} with the observation that a single round of PC can be simulated in a single round of the CLIQUE model, we obtain the following.

\begin{corollary}\label{cor:clique-ec}
	Suppose a computational task as complexity $T_C(n, r)$ in the CLIQUE model for hypergraphs with $n$ vertices and rank $r$. Then its computational complexity $T_{EC}(n, r)$ in the EC model satisfies
	\begin{equation}
		T_{EC}(n, r) = \Omega\paren{\frac{1}{n^{r-2}} T_C(n, r)}.
	\end{equation}
\end{corollary}
\begin{proof}
	The conclusion follows from Proposition~\ref{prop:pc-relationship} together with the observation that a single round of PC can be simulated in a single round of the CLIQUE model, and the fact that for any hypergraph $H$ with rank $r$ and $n$ vertices, we have $\Delta_2 \leq \sum_{k = 0}^{r-2} \binom{n-2}{k} = O(n^{r-2})$.
\end{proof}



\subsection{Obtaining Small Unique Edge IDs}\label{sec:edge-ids}

We will now briefly describe a scheme to assign short, unique IDs to edges in addition to unique vertex IDs.
Below, we describe a scheme by which the vertices can assign unique IDs to edges in $O(1)$ rounds of the EC, EB, and ES models, $O(\log r)$ rounds in the EU model and $O(r)$ rounds in the EP model, assuming each vertex knows the identities of all vertices in each of its incident edges (i.e. the $\KT_1$ assumption above). Then we can assign unique IDs of size $\log n$ as follows:

\begin{itemize}
	\item For an edge $e$, let $v$ be the vertex in $e$ with minimal ID, and let $p$ be the port number that $v$ associates with $e$. That is, $e$ is $v$'s $p$-th incident edge in some arbitrary (but fixed) ordering of $v$'s incident edges. Then assign the ID $(v, p)$ to $e$.
	\item Since other vertices $w \in e$ need not know $p$, this port number can be disseminated to all vertices in $e$ in a single round of the ES model, and hence in $O(1)$ rounds of EC, EB, $O(\log r)$ rounds of EU, and $O(r)$ rounds of EP.
\end{itemize}

%% file: clique.tex
\section{Simplex and Triangle Enumeration in the CLIQUE Model}
\label{sec:clique}

In this section, we consider the general \dft{simplex enumeration} and more restricted \dft{triangle enumeration} problems in the CLIQUE model. In this task, each vertex $v$ should output a set of triangles (cycles of length 3), $S_v$, such that the $S_v$ are pair-wise disjoint and $S = \bigcup_{v \in V} S_v$ contains all triangles in $H$.

%% file: clique-ub.tex
\subsection{Simplex and Triangle Enumeration Algorithm}
\label{sec:clique-ub}

In this section, we prove our main upper bound result for simplex and triangle enumeration in the CLIQUE model. Throughout the section, we assume that $H = (V, E)$ is a hypergraph with $V = [n]$ and that the rank of $H$ is $r$. Our main results are as follows:

\begin{theorem}\label{thm:s-clique-ub}
	For any fixed constant $r$, the $r$-dimensional simplex enumeration problem can be solved in $O(n^{r-2+1/(r+1)}/\log n)$ rounds in the CLIQUE model for $r$-uniform hypergraphs.
\end{theorem}

\begin{theorem}\label{thm:clique-ub}
	For any fixed constant $r$, the triangle enumeration problem can be solved in $O(n^{r-5/3}/\log n)$ rounds in the CLIQUE model for hypergraphs of rank $r$.
\end{theorem}

Note that for $r = 2$ (i.e., for graphs), both of these upper bounds are $O(n^{1/3}/\log n)$, matching the bound of Dolev, Lenzen, and Peled~\cite{Dolev2012Tri} for triangle enumeration in graphs. The algorithms that achieve the upper bounds of Theorems~\ref{thm:s-clique-ub} and~\ref{thm:clique-ub} are generalizations of the algorithm of~\cite{Dolev2012Tri}. The algorithm achieving Theorem~\ref{thm:s-clique-ub} is slightly more general, so we provide full details of that algorithm here, and describe how the procedure can be modified to obtain Theorem~\ref{thm:clique-ub} as well. Throughout the section we use the term \dft{$k$-subset} to refer to a subset of $[n]$ of size $k$, and denote the set of all $k$-subsets of $[n]$ by $\binom{[n]}{k}$.

The basic idea for the algorithm achieving Theorem~\ref{thm:s-clique-ub} is to partition the set of all $(r+1)$-subsets $\set{v_0, v_1, \ldots, v_r} \subseteq V$, $\binom{V}{r+1}$, into $n$ sets, $T_1, T_2, \ldots, T_n$ corresponding to the $n$ vertices in the hypergraph. Each vertex $i \in V$ is then responsible for outputting all simplices $\Sigma$ such that $\Sigma$ is an $(r+1)$-set in $i$'s part of the partition, $S_i$. In order for $i$ to output all such simplices, all edges $e$ that are subsets of some set $\Sigma \in S_i$ should be sent to $i$. Dolev et al.~\cite{Dolev2012Tri} gave an elegant description of such a partition of triples $r=2$ case such that total number of edges sent and received by each vertex is (an asymptotically optimal) $\Theta(n^{4/3})$ for graphs (i.e., rank-$2$ hypergraphs). We show that the same strategy generalizes to $(r+1)$-sets also yields an optimal algorithm for simplex enumeration in hypergraphs.

To describe the partition of $(r+1)$-sets, first partition $V$ into $p = n^{1/(r+1)}$ sets of size $q = n/p = n^{r/{r+1}}$, $S_1, S_2, \ldots, S_p$. (For notational convenience, we assume that $n$ is a power of $r+1$ so that $s = n^{1/(r+1)}$ is an integer.) For each integer $i \in [n]$, write $i = (i_r i_{r-1} \cdots i_0)_p$ in base $p$ where each $i_j \in [p]$.

The set of $(r+1)$-tuples associated with vertex $i = (i_r i_{r-1}\cdots i_0)_p$ is then the set $T_i$ consisting of all $(r+1)$-tuples of the form $S_{i_r} \times S_{i_{r-1}} \times \cdots \times S_{i_0}$.\footnote{Technically, the $T_i$ form a a partition of the set of (ordered) $r+1$ tuples, and not $(r+1)$-subsets of $[n]$ as suggested above. We can easily change such a partition into a partition of $(r+1)$-subsets by associating each $(r+1)$-subset $\Sigma$ with the lexicographically first $(r+1)$-tuple containing the elements of $\Sigma$. The vertex assigned this tuple is then responsible for outputting the simplex $\Sigma$, if present in $H$.} That is, we define
\begin{equation}
	T_i = S_{i_r} \times S_{i_{r-1}} \times \cdots \times S_{i_0}.
\end{equation}
For a set $T$ of tuples of elements from $[n]$, we define the \dft{$r$-support} of $T$, denoted $\supp_r(T)$, to be the set of $r$-subsets of $[n]$ that contains all $r$-subsets of the support of some tuple in $T_i$.\footnote{Recall that given a tuple $s = (s_r, s_{r-1}, \ldots, s_{0})$, the \dft{support} of $s$, denoted, $\supp(s) = \set{s_r, s_{r-1},\ldots, s_0}$, is the set of distinct elements appearing in $s$.}
\begin{equation}\label{eqn:e-i}
	\supp_r(T) = \set{f \in \binom{[n]}{r} \sucht \exists s \in T_i, f \subseteq \supp(s).}
\end{equation}
The set $F_i = \supp_r(T_i)$ corresponds to all possible rank $r$ edges whose elements appear in an $(r+1)$-tuple in $T_i$. For a hypergraph $H = (V, E)$, we define $E_i = E \cap F_i$ to be the actual edges present in $H$ associated in this way with $T_i$. We make the following observations regarding the definitions above.

\begin{observation}\label{obs:triangle-partition}
	Let $H = (V, E)$ be a hypergraph with $V = [n]$ and suppose $T_i$ and $E_i$ are defined as above. Then:
	\begin{enumerate}
		\item $T_1, T_2, \ldots, T_n$ is a partition of $V^{r+1}$.
		\item If $\Sigma = \set{v_0, v_1, \ldots, v_r}$ is a simplex in $H$ with $(v_r, v_{r-1}, \ldots, v_0) \in T_i$ then for every $r$-subset $e$ of $\Sigma$, we have $e \in E_i$.
	\end{enumerate}
\end{observation}

By the observation, in order to perform simplex enumeration, it suffices for each vertex $i$ to learn the set of edges $E_i$. The main algorithm is simply a scheme for communicating edges such that each $i$ receives all edges in $E_i$. Before describing the algorithm formally, we make some more observations about the partition $T_1,\ldots,T_n$ as well as the associated set of possible edges $F_i$.

\begin{lemma}\label{lem:simplex-support}
	Let $H = (V, E)$ be an $r$-uniform hypergraph for some constant $r \geq 2$ with $V = [n]$. Define $\calT = \set{T_1, T_2, \ldots, T_n}$ to be the partition of $(r+1)$-tuples defined above, and for each $i \in [n]$, define $F_i = \supp_r(T_i)$ and $E_i = E \cap F_i$. Then the following hold:
	\begin{enumerate}
		\item For for each $i$, we have $\abs{E_i} = O(n^{r - 1 + 1/(r+1)})$.
		\item For each edge $e \in E$, there are $O(n^{1/(r+1)})$ indices $i$ for which $e \in E_i$.
		\item For each vertex $v$, there are $O(n^{r/(r+1)})$ indices $i$ for which there exists an incident edge $e \ni v$ with $e \in E_i$.
		\item For any vertex $v$ and index $i$, the number of edges $e$ incident to $v$ for which $e \in E_i$ is $O(n^{r - 2 + 2/(r+1)})$
	\end{enumerate}
\end{lemma}
\begin{proof}
	We prove each of the claims separately.
	\begin{enumerate}
		\item Since $E_i \subseteq F_i$, it suffices to show the claimed bound for $\abs{F_i}$. To this end, recall that $T_i = S_{i_r} \times S_{i_{r-1}} \times \cdots \times S_{i_0}$ where each $S_{i_j}$ satisfies $\abs{S_{i_j}} = q = \Theta(n^{r/(r+1)})$. Observe that each set $e \in F_i$ is an $r$-subset of $\bigcup_{j = 0}^r S_{i_{j}}$, which has size $(r+1)q$. Thus, the number of sets $e \in F_i$ is at most $\binom{(r+1)q}{r} \leq (r+1)^r q^r = \Theta(n^{r^2/(r+1)}) = \Theta(n^{r-1+1/(r+1)})$.
		\item Let $e = \set{v_1, v_2, \ldots, v_r}$. Note that $e \in \supp(T_i)$ if and only if there is an injective map $\varphi : [r] \to [r+1]$ for which for all $j$, $v_j \in S_{i_{\varphi(j)}}$. There are at most $(r+1)!p = O(n^{1/(r+1)}$ indices $i$ for which this condition is satisfied. To see, this, observe that for any fixed function $\varphi$, there are $p$ indices $i$ for which $v_j \in S_{i_{\varphi(j)}}$ for all $j$, corresponding to the $p$ choices of $S_{i_{k}}$ where $k$ is the (unique) index not in the range of $\varphi$.
		\item Since every $e$ incident to $v$ contains $v$ as an element, every $i = (i_r i_{r-1} \cdots i_0)_p$ must contain $i_{v}$ as an index, where $v \in S_{i_v}$. There are $r+1$ indices $j$ for which we could have $i_v = i_j$, $p$ choices of each the remaining indices $i_k$, $k \neq j$. Thus there are at most $(r + 1) p^r = O(n^{r/(r+1)})$ possible indices $i$ for which $v$ is incident to an edge $e \in E_i$.
		\item Let $v \in V$ and consider the set $F_{i, v} = \set{e \in \supp_r(T_i) \sucht v \in e}$. We must show that $\abs{F_{i, v}} = O(n^{r - 2 + 2/(r+1)})$. To this end, observe that for the tuples $(i_r, i_{r-1},\ldots,i_0)$, there are at most $r+1$ indices $j$ for which we could have $v = i_j$ and a remaining $r$ indices $k \neq j$ for which $i_k$ is omitted from $e$ (the $r$-set containing $v$). For each the remaining values $i_\ell$, $\ell \neq j, k$, there are $\abs{S_{i_\ell}} = q$ possible choices. Thus, the total number of edges in $F_{i, v}$ is at most $r (r+1) q^{r-1} = O(n^{r(r-1)/(r+1)}) = O(n^{r - 2 + 2 / (r+1)})$.
	\end{enumerate}
	Thus, all the claims of the lemma hold.
\end{proof}

Using the bounds from Lemma~\ref{lem:simplex-support}, we can describe a simple algorithm for simplex enumeration in the CLIQUE model that uses $O(n^{r-2+1/(r+1)})$ rounds as follows. Again, the idea is that for each pair of vertices $v$ and $i$, $v$ should send each incident edge $e$ to the vertex $i$ for which $e \in F_i$. By Claim~1 in the lemma, each $i$ will receive at most $O(n^{r-1+1/(r+1)})$ edges, each of which can be encoded with $r \log n = O(\log n)$ bits. By Claim~2 in the lemma, each edge $e$ incident to $v$ will be sent to $O(n^{1/(r+1)})$ other vertices $i$. Since $v$ has degree at most $\binom{n-1}{r-1} = O(n^{r-1})$, the total number messages sent by $v$ is $O(n^{r-1+1/(r+1)})$ and each message has length $O(\log n)$. Under this scheme, each vertex is the source and destination of $O(n^{r-1+1/(r+1)})$ messages, hence using Lenzen's deterministic routing algorithm~\cite{Lenzen2013Optimal}, all messages can be routed to their destinations using in $O(n^{r-2+1/(r+1)})$ rounds of CLIQUE.

In order to improve this running time by a factor of $1/\log(n)$, we employ the following technique (also used by Dolev et al.~\cite{Dolev2012Tri}). Instead of sending individual messages for each incident edge, each vertex forms bit vectors that encode the present edges using only one bit per edge. We formalize this procedure in the proof below.

\begin{proof}[Proof of Theorem~\ref{thm:s-clique-ub}]
	For each pair of vertices $v$ and $i$, let $F_{v, i} = \set{e \in F_i \sucht v \in e}$. That is, $F_{v, i}$ contains all $r$-subsets of $F_i$ that contain $v$. The vertex $v$ forms the bit vector $b_{v, i} \in \set{0, 1}^{\abs{F_{v, i}}}$ indexed by elements from $F_{v, i}$ as follows:
	\begin{equation*}
		b_{v, i}(e) =
		\begin{cases}
			1 & \text{if } e \in E \\
			0 & \text{otherwise.}
		\end{cases}
	\end{equation*}
	By Claim~3 of Lemma~\ref{lem:simplex-support}, for each $v$, there are only $O(n^{r/(r+1)})$ vertices $i$ for which $F_{v, i}$ is nonempty. Further, by Claim~4, each $b_{v, i}$ has length $O(n^{r-2+2/(r+1)})$. By breaking each $b_{v, i}$ up into $n^{1/(r+1)}$ messages each of length $O(n^{r-2+1/(r+1)})$, $v$ is the source of $O(n)$ messages each of length $O(n^{r-2+1/(r+1)})$. Symmetrically, each $i$ is the destination of $O(n)$ messages each of length $O(n^{r-2+1/(r+1)})$.

	Since each vertex is the source and destination of $O(n)$ messages each of length $O(n^{r-2+1/(r+1)})$, all messages can be routed deterministically using a straightforward generalization of Lenzen's routing~\cite{Lenzen2013Optimal} in $O(n^{r-2+1/(r+1)}/\log n)$ rounds of the CLIQUE model.\footnote{Lenzen~\cite{Lenzen2013Optimal} showed that if each vertex is the source and destination of at most $n$ messages, each of which can be transmitted across a single edge in a single round, then all messages can be routed from source to destination in $O(1)$ rounds deterministically. This result can be adapted to our setting (where individual messages cannot be transmitted in a single round) in a straightforward way. Specifically, we break an execution up into phases, where each phase simulates an execution of Lenzen's algorithm. Whenever Lenzen's algorithm transmits a message $M$ from a vertex $v$ to a vertex $u$, we send the message in $O(\abs{M}/\log n)$ rounds of CLIQUE. Thus, the total number of rounds to route messages with each vertex the source and destination of at most $n$ messages is $O(\abs{M}/\log n)$.} Once each vertex $i$ has received all such messages, it can reconstruct $E_i$, from which $i$ can output all present simplices from $T_i$, as required.
\end{proof}

The algorithm and argument achieving Theorem~\ref{thm:clique-ub} (triangle enumeration) is a straightforward adaption of our proof of Theorem~\ref{thm:s-clique-ub}, so we only sketch how the procedure and analysis differ here.

\begin{proof}[Proof of Theorem~\ref{thm:clique-ub} (sketch)]
	In order to perform triangle enumeration in a hypergraph $H = (V, E)$, we once again assume $V = [n]$. The idea is to partition the set of all triples $(u, v, w) \in V^3$ into $n$ sets using the same partition $\set{T_1, T_2, \ldots, T_n}$ as in the proof of Theorem~\ref{thm:s-clique-ub} with $r = 2$. The idea is that vertex $i$ will be responsible for outputting all triangles $(v_0, e_0, v_1, e_1, v_2, e_2, v_0)$ such that $(v_0, v_1, v_2) \in T_i$. To accomplish this, each vertex $v \in \set{v_0, v_1, v_2}$ should send to $i$ all of $v$'s incident edges containing at least one other element $w$ from $\set{v_0, v_1, v_2}$. Each pair $\set{v_i, v_j}$ is contained at most $\binom{n-2}{r-2} = O(n^{r-2})$ edges, and by Lemma~\ref{lem:simplex-support} (Claim~1, with $r=2$), the triples in $T_i$ in $T_i$ contain $O(n^{4/3})$ such pairs in their support. Thus, each $i$ receives $O(n^{r-2/3})$ total edges. Similarly, by Claim~2 of Lemma~\ref{lem:simplex-support} (again with $r = 2$), $v$ sends each incident edge to $O(n^{1/3})$ vertices $i$, so the total number edges sent by each $v$ is $O(n^{r-2/3})$, as $v$ has at most $n$ neighboring vertices, and for each neighboring vertex $w$ there are $O(n^{r-2})$ edges of rank at most $r$ containing both $v$ and $w$. Since each vertex sends and receives $O(n^{r-2/3})$ edges, we can apply bit vector encoding as in the proof of Theorem~\ref{thm:s-clique-ub} together with Lenzen routing~\cite{Lenzen2013Optimal} to deterministically list triangles in $O(n^{r-5/3}/\log n)$ rounds of CLIQUE.
\end{proof}

%% file: clique-lb.tex
\subsection{Simplex and Triangle Enumeration Lower Bound}
\label{sec:clique-lb}

In this section, we prove our main lower bounds for the simplex and triangle enumeration problems, thereby showing that the upper bounds of Theorems~\ref{thm:s-clique-ub} and~\ref{thm:clique-ub} are tight. The primary result generalizes the CLIQUE lower bound of Izumi and Le Gall~\cite{Izumi2017Triangle} to hypergraphs.

\begin{theorem}\label{thm:clique-lb}
	For any constant $r$ there exists a distribution over $r$-uniform hypergraphs such that:
	\begin{enumerate}
		\item any algorithm solving $r$-dimensional simplex enumeration with sufficiently large constant probability requires $\Omega(n^{r-2+1/(r+1)}/\log n)$ rounds in the CLIQUE model, and
		\item any algorithm solving triangle enumeration with sufficiently large constant probability requires $\Omega(n^{r-5/3} / \log n)$ rounds in the CLIQUE model.
	\end{enumerate}
	These bounds applies \emph{a fortiori} to the PRIMAL CONGEST (PC) model as well.
\end{theorem}

Applying Theorem~\ref{thm:clique-lb} along with Corollary~\ref{cor:clique-ec} and Proposition~\ref{prop:model-relationships}, we obtain the following corollary for simplex and triangle enumeration in the various models introduced in Section~\ref{sec:computational-models}.

\begin{corollary}\label{cor:congest-lb}
	Solving simplex enumeration and (respectively) triangle enumeration with sufficiently large probability in the EC, EB, EU, ES, and EP models in hypergraphs with $n$ vertices require $\Omega(n^{1/(r+1)}/\log n)$ and $\Omega(n^{1/3}/\log n)$ rounds, respectively.
\end{corollary}

Our proof of Theorem~\ref{thm:clique-lb} generalizes the lower bound argument of Izumi and Le Gall to hypergraphs. The idea of Izumi and Le Gall's proof is to consider a random input drawn from $G(n, 1/2)$, where each pair of vertices forms an edge independently with probability $1/2$. A graph $G$ sampled from $G(n, 1/2)$ contains $\Omega(n^3)$ triangles in expectation. Thus, in order to solve the triangle enumeration problem, some vertex must output $\Omega(n^2)$ triangles (with constant probability).

Izumi and Le Gall's argument relies upon two main lemmas. The first (Lemma~4.3 in~\cite{Izumi2017Triangle}) gives a lower bound on the mutual information between the set of triangles output by a vertex and the edges of the graph by the number of edges contained in the triangles output by the vertex. The second lemma (Lemma~4.2 in~\cite{Izumi2017Triangle}) gives a lower bound on the number of edges contained in $t$ triangle, namely any set of $t$ triangles contains at most $t^{2/3}$ edges.

Izumi and Le Gall's main argument then proceeds as follows. For $G \sim G(n, 1/2)$, some vertex $v$ outputs $\Omega(n^2)$ triangles with constant probability. By the first lemma (Lemma~4.3), the mutual information between the edge set of $G$ and the set $T$ of triangles is at least the number of edges contained in some triangle of $T$. This quantity is $\Omega(n^{4/3})$ by the second lemma (Lemma~4.2). The lower bound then follows by arguing that mutual information between $E$ and $T$ is a lower bound for the number of bits received by $v$. Since $v$ can receive at most $n \log n$ bits per round, the lower bound $n^{1/3} / \log n$ follows.

To prove Theorem~\ref{thm:clique-lb}, we generalize each of the parts of Izumi and Le Gall's argument, adapting it both to simplex enumeration and hypergraph triangle enumeration. Specifically, our argument for simplex enumeration proceeds as follows:
\begin{enumerate}
	\item We introduce the random hypergraph model $H(n, r, p)$ that generalizes the random graph model $G(n, p)$, and we show that $H \sim H(n, r, 1/2)$ contains $\Theta(n^{r+1})$ $r$-dimensional simplices with constant probability. Hence, in order to solve simplex enumeration, some vertex must output $\Omega(n^r)$ simplices with constant probability. In the case of triangles, we show that $H$ contains $\Theta(n^{3r - 3})$ triangles with constant probability, hence some vertex must output $\Theta(n^{3r - 4})$ triangles with constant probability.
	\item Next we show that that any sub-hypergraph containing many simplices (respectively triangles) must contain many edges. In the case of simplices, we apply the Kruskal-Katona theorem from extremal combinatorics to this end. For triangles, we show that any set of $m$ edges in an $r$-uniform hypergraph on $n$ vertices that contains $t$ triangles must satisfy $m = \Omega(t^{2/3}/n^{r - 2})$. This result generalizes the result of Rivlin~\cite{Rivin2002Counting} relating the minimum number of edges needed to produce a given number of triangles in a graph.
	\item Finally, we apply the information theoretic framework of Izumi and Le Gall~\cite{Izumi2017Triangle} to infer the lower bounds of Theorem~\ref{thm:clique-lb} from the combinatorial bounds above.
\end{enumerate}

Consider the random $r$-uniform graph model $H(n, r, p)$, where each set in $\binom{[n]}{r}$ is included as an edge in $G$ independently with probability $p$. The following lemmas give bounds on the number of triangles contained in $H(n, r, 1/2)$.

\begin{lemma}\label{lem:expected-simplices}
	Suppose $H \sim H(n, r, 1/2)$, and let $S$ be the random variable for the number of $r$-dimensional simplices in $H$. Then
	\begin{equation}\label{eqn:esimplex-lb}
		\E(S) = \frac{1}{2^{r+1}} \binom{n}{r+1} .
	\end{equation}
	Thus, for constant $r$, we have $\E(S) = \Omega(n^{r + 1})$. Further, we have
	\begin{equation}\label{eqn:simplex-ub}
		S \leq \binom{n}{r+1} = \Theta(n^{r+1}).
	\end{equation}
\end{lemma}
\begin{proof}
	For any $(r+1)$-subset $\Sigma \in \binom{[n]}{r+1}$, $\Sigma$ forms a simplex in $H$ if and only if all of the sets $e \in \binom{\Sigma}{r}$ are included as edges in $H$. Since there are $r+1$ such sets, this occurs with probability $2^{-r-1}$. Equation~(\ref{eqn:esimplex-lb}) then follows from linearity of expectation. The upper bound of~(\ref{eqn:simplex-ub}) simply bounds the number of simplices by the number of distinct $(r+1)$-subsets of $[n]$.
\end{proof}

\begin{lemma}\label{lem:expected-triangles}
	Suppose $H \sim H(n, r, 1/2)$, and let $T$ be the random variable for the number of triangles contained in $H$. Then
	\begin{equation}\label{eqn:etriangle-lb}
		\E(T) \geq \frac{1}{8} \binom{n}{3} \binom{n-3}{r-2} \binom{n-r-1}{r-2} \binom{n-2r+1}{r-2}.
	\end{equation}
	Thus, for constant $r$, we have $\E(T) = \Omega(n^{3r - 3})$. Further, we have
	\begin{equation}\label{eqn:triangle-ub}
		T \leq \binom{n}{3} \binom{n-3}{r-2}^3 < n^{3r - 3}.
	\end{equation}
\end{lemma}
\begin{proof}
	Fix any three vertices $u, v, w \in V$ with $u < v < w$, so that there are $\binom{n}{3}$ such choices of $(u, v, w)$. To form a triangle with these three vertices, there must be distinct edges containing each pair of vertices, but not containing all three vertices. There are $\binom{n-3}{r-2}$ sets $e_{uv}$ in $\binom{[n]}{r}$ containing $u$ and $v$ but not $w$. Similarly there are $\binom{n-r-1}{r-2}$ remaining $r$-sets $e_{vw}$ containing $v$ and $w$ but not containing $u$ or elements from $e_{uv}$, and $\binom{n-2r+1}{r-2}$ remaining $r$ sets $e_{wu}$ containing $u$ and $w$ but not $v$ or elements from $e_{uv}$ or $e_{vw}$. Thus, $\tau = u e_{uv} v e_{vw} w e_{wu} u$ forms a triangle in $G$ if and only if the three edges are chosen. Since each of these edges is chosen independently with probability $1/2$, $\tau$ forms a triangle in $H$ with probability $1/8$. Summing over all all triples $(u, v, w)$ and possible edges $e_{uv}, e_{vw}, e_{wu}$ and applying linearity of expectation, we obtain Equation~\ref{eqn:etriangle-lb}).

	For Equation~(\ref{eqn:triangle-ub}), observe that each possible triangle consists of three distinct vertices and three edges, each containing precisely two of the three vertices. There are $\binom{n}{3}$ such sets of three vertices and $\binom{n-3}{r-2}$ possible edges containing exactly two vertices. Thus there are at most $\binom{n}{3} \binom{n-3}{r-2}^3$ configurations corresponding to triangles.
\end{proof}

\begin{corollary}\label{cor:expected-triangles}
	If $H \sim H(n, r, 1/2)$, then $H$ contains $\Omega(n^{r+1})$ $r$-dimensional simplices and $\Omega(n^{3r - 3})$ triangles with constant probability.
\end{corollary}

\begin{proof}
	By Lemma~\ref{lem:expected-triangles}, there exists an absolute constant $c$ such that $\E(T) \geq c n^{3r - 3}$. On the other hand, $T < n^{3r - 3}$ by Equation~(\ref{eqn:triangle-ub}). For $\alpha$ satisfying $0 < \alpha < 1$, let $A$ be the event that $T < \alpha n^{3r - 3}$, and take $p_\alpha = \Pr(A)$. Then we have
	\[
		c n^{3r - 3} \leq \E[T] \leq \alpha n^{3r - 3} p_\alpha + (1 - p_\alpha) n^{3r - 3}.
	\]
	hence
	\[
		p_\alpha \geq \frac{c - \alpha}{1 - \alpha}.
	\]
	Taking $\alpha = c/2$ gives, $p_\alpha = \Omega(1)$ as required. A similar argument proves the bound for the case of simplices.
\end{proof}

Applying the pigeonhole principle to Corollary~\ref{cor:expected-triangles}, we immediately obtain the following.

\begin{corollary}\label{cor:triangle-cover}
	Suppose $H \sim H(n, r, 1/2)$ and let $\calS$ and $\calT$ denote the sets of $r$-dimensional simplices and triangles (respectively) in $H$. Suppose $\calS_1, \calS_2, \ldots, \calS_n$ satisfy $\calS = \bigcup_{i=1}^n \calS_i$, and similarly  $\calT_1, \calT_2, \ldots, \calT_n$ satisfy $\calT = \bigcup_{i = 1}^n \calT_i$. Then with constant probability $\max_i \abs{S_i} = \Omega(n^r)$ and $\max_i \abs{T_i} = \Omega(n^{3r - 4})$.
\end{corollary}

By Corollary~\ref{cor:triangle-cover}, an algorithm that solves the simplex or triangle listing problem (in the congested clique) has the property that if the input $H$ is sampled from $H(r, n, 1/2)$, then with constant probability, some vertex $v$ outputs $\Omega(n^r)$ simplices or $\Omega(n^{3r - 4})$ triangles. Towards proving Theorem~\ref{thm:clique-lb}, our next step is to argue that such a set of simplices or triangles must be supported by a large number of edges. That is, some vertex $v$ must witness many edges in $H$. To this end, we require two separate results for simplices and triangles. In the case of simplices, the well-known Kruskal-Katona theorem from extremal combinatorics suffices. We state the version of the theorem proven in Jukna's monograph~\cite{Jukna2011Extremal} here.

\begin{theorem}[Kruskal-Katona, c.f.~{\cite[Theorem 10.15]{Jukna2011Extremal}}]\label{thm:kruskal-katona}
	Suppose $H = (V, E)$ is an $r$-uniform hypergraph and let $\calS$ denote the set of simplices in $H$. If
	\[
		\abs{\calS} \geq \binom{x}{r+1} = \frac{1}{(r+1)!} x (x-1)\cdots(k-r)
	\]
	for some positive real number $x \geq r+1$, then
	\begin{equation}\label{eqn:kruzkal-katona}
		\abs{E} \geq \binom{x}{r}.
	\end{equation}
\end{theorem}

Applying this result in our setting, we obtain the following corollary showing that the number edges contained in the support of $\Theta(n^r)$ simplices is large.

\begin{corollary}\label{cor:simplex-edge-bound}
	Let $H$ be an $r$-uniform (sub)hypergraph with $n$ vertices and $m$ edges containing $S = \Omega(n^{r})$ simplices. Then $m = \Omega(n^{r-1+1/(r+1)})$.
\end{corollary}
\begin{proof}
	Suppose $x$ is a real number satisfying $\abs{\calS} = \binom{x}{r+1}$. We can find such an $x$ because $\binom{x}{r+1}$ is a continuous unbounded function of $x$. Observe that
	\begin{equation}
		\binom{x}{r+1} < \frac{x^{r+1}}{(r+1)!}.
	\end{equation}
	Since $\binom{x}{r+1} = \abs{\calS} = \Omega(n^r)$, applying the previous expression gives
	\begin{equation*}
		\frac{x^{r+1}}{(r+1)!} = \Omega(n^r), \quad\text{hence}\quad x = \Omega(n^{r/(r+1)}).
	\end{equation*}
	We then bound the following:
	\begin{align*}
		\abs{E} & \geq \binom{x}{r}
		\geq \paren{\frac{x}{r}}^r
		= \Omega\paren{\frac{1}{r^r} n^{r^2/(r+1)}}
		= \Omega(n^{r-1+1/(r+1)}),
	\end{align*}
	where the final equality holds for constant $r$.
\end{proof}

For the case of triangles, we prove the following combinatorial lower bound on the number of edges in the support of a large number of triangles.

\begin{theorem}\label{thm:triangle-edge-bound}
	Suppose $H = (V, E)$ is an $r$-uniform hypergraph that contains $n$ vertices, $m$ edges, and $t$ triangles. Then
	\begin{equation}\label{eqn:triangle-edge-bound}
		t \leq \frac{\sqrt{2}}{3} \paren{m \binom{n-3}{r-2}\binom{r}{2}}^{3/2} .
	\end{equation}
	Solving for $m$, this gives $m = \Omega(t^{2/3}/r^2 n^{r-2})$.
\end{theorem}

Applying Theorem~\ref{thm:triangle-edge-bound} for a set of $t = \Omega(n^{3r -4})$ triangles (as in Corollary~\ref{cor:triangle-cover}), we obtain the following lower bound on the number of edges.

\begin{corollary}\label{cor:triangle-edge-bound}
	Let $H$ be an $r$-uniform (sub)hypergraph with $n$ vertices and $m$ edges containing $t = \Omega(n^{3r-4})$ triangles. Then $m = \Omega(n^{r-2/3}/r^2)$.
\end{corollary}

In order to prove Theorem~\ref{thm:triangle-edge-bound}, we generalize the techniques of Rivin~\cite{Rivin2002Counting}, who's result corresponds to Theorem~\ref{thm:triangle-edge-bound} in the case $r = 2$ (i.e., graphs). Rivin's analysis considers the the adjacency matrix $A$ of a graph $G$ and proves a relationship between the trace $\tr(A^2)$ (which is proportional to the number of edges of $G$) and $\tr(A^3)$ (which is proportional to the number of triangles in $G$).

In order to apply Rivin's technique to hypergraphs, we define a multigraph $G = G_{\set{H}}$ whose triangles are in one-to-one correspondence with the triangles of $H$, then generalize Rivin's result to multigraphs of bounded edge multiplicity.

\begin{definition}\label{dfn:induced-multigraph}
	Let $H = (V, E)$ be a hypergraph. We define the \dft{induced multigraph} $G_{\set{H}} = (V, E')$ to be the multigraph with vertex $V$ and edge set
	\[
		E' = \mset{\set{u, v} \sucht \exists e \in E, u, v \in e}.
	\]
	That is, $E'$ contains the multiset of pairs of vertices contained in all edges $E$ in $H$.
\end{definition}

\begin{lemma}\label{lem:edge-multiplicity}
	Suppose $H = (V, E)$ is an $r$-uniform hypergraph with $m$ edges. Then $G_{\set{H}}$ contains $m \binom{r}{2}$ edges, and each edge occurs with multiplicity at most $\binom{n-2}{r-2} = O(n^{r-2})$.
\end{lemma}
\begin{proof}
	For the first claim, observe that each edge $e \in E$ induces $\binom{\abs{e}}{2}$ edges in $G_{\set{H}}$. For the second claim, observe that a given pair $\set{u, v}$ can occur in at most $\binom{n-2}{r-2}$ edges in $E$.
\end{proof}

\begin{lemma}\label{lem:triangle-correspondence}
	There is a one-to-one correspondence between (hypergraph) triangles in $H$ and (graph) triangles in $G_{\set{H}}$.
\end{lemma}
\begin{proof}
	For each triangle $T = (v_0, e_0, v_1, e_1, v_2, e_2)$ in $H$, consider the corresponding triangle $T' = (v_0, e_0', v_1, e_1', v_2, e_2')$ where $e_i' = \set{v_i, v_{i+1}}$ is the edge induced by $e_i$ in $G_{\set{H}}$. It is clear that $T \mapsto T'$ is a one-to-one correspondence between triangles in $H$ and $G_{\set{H}}$.
\end{proof}

By Lemmas~\ref{lem:edge-multiplicity} and~\ref{lem:triangle-correspondence}, in order to prove Theorem~\ref{thm:triangle-edge-bound}, it suffices to prove an analogous result for multigraphs. To this end, we first introduce the adjacency matrix of a multigraph.

\begin{definition}\label{dfn:multigraph-adjacency}
	Let $G = (V, E)$ be a multigraph on $n$ vertices, $V = \set{v_1, v_2,\ldots,v_n}$. Then the adjacency matrix $A = (a_{ij})$ of $G$ is the (real, symmetric) $n \times n$ matrix where $a_{ij}$ is the number of edges containing the vertices $v_i$ and $v_j$.
\end{definition}

Just as with adjacency matrices of simple graphs, powers of the multigraph adjacency matrix count paths of a fixed length between pairs of vertices. Specifically, we obtain the following interpretation:

\begin{observation}\label{obs:matrix-paths}
	Suppose $A = (a_{ij})$ is the adjacency matrix of a multigraph $G = (V, E)$ and $B = (b_{ij}) = A^k$. Then $b_{ij}$ counts the number of distinct paths of length $k$ from $v_i$ to $v_j$.
\end{observation}

\begin{lemma}\label{lem:matrix-edge-triangle}
	Suppose $G = (V, E)$ is a multigraph with $n$ vertices, $m$ edges, and $t$ triangles where the multiplicity of each edge is at most $\mu$. Let $A = (a_{ij})$ denote the adjacency matrix of $G$. Then,
	\begin{align}
		\tr(A^2) & \leq 2 \mu m, \quad\text{and}\label{eqn:m-bound} \\
		\tr(A^3) & = 6 t.\label{eqn:t-bound}
	\end{align}
\end{lemma}
\begin{proof}
	In order to prove~(\ref{eqn:m-bound}), we compute
	\[
		(A^2)_{ii} = \sum_{j = 1}^n a_{ij} a_{ji} \leq \sum_{j = 1}^n a_{ij} \mu = \mu \deg(v_i).
	\]
	Thus,
	\[
		\tr(A^2) = \sum_i (A^2)_{ii} \leq \sum_i \mu \deg(v_i) = 2 \mu m,
	\]
	as desired. For~(\ref{eqn:t-bound}), applying Observation~\ref{obs:matrix-paths}, we find that $(A^3)_{ii}$ is the number of 3-paths that start and end at vertex $v_i$. This value is twice the number of triangles containing $v_i$, as each triangle is counted twice (corresponding to the two directions of traversing the triangle). $\tr(A^3) = \sum_{i} (A^3)_{ii}$ therefore counts each triangle $6$ times, as each triangle is counted twice by each of the three vertices in the triangle.
\end{proof}

Finally, in order to prove Theorem~\ref{thm:triangle-edge-bound}, we apply the following lemma of Rivin.

\begin{lemma}[{\cite[Corollary 1]{Rivin2002Counting}}]\label{lem:trace-bound}
	Let $A$ be a symmetric $n \times n$ matrix. Then
	\begin{equation}\label{eqn:trace-bound}
		(\tr(A^k))^2 \leq (\tr(A^2))^k.
	\end{equation}
\end{lemma}

Combining the previous two lemmas, we obtain the following.

\begin{corollary}\label{cor:multigraph-triangle-bound}
	Let $G = (V, E)$ be a multigraph with $n$ vertices, $m$ edges and $t$ triangles. Suppose each edge in $E$ has multiplicity at most $\mu$. Then
	\begin{equation}\label{eqn:multigraph-triangle-bound}
		(6 t)^2 \leq (2 \mu m)^3.
	\end{equation}
\end{corollary}

We now have all of the pieces in place to prove Theorem~\ref{thm:triangle-edge-bound}.

\begin{proof}[Proof of Theorem~\ref{thm:triangle-edge-bound}]
	Let $H = (V, E)$ be an $r$-uniform hypergraph with $n$ vertices, $m$ edges, and $t$ triangles. Let $G_{\set{H}} = (V, E')$ be its induced multigraph and suppose $A$ is $G_{\set{H}}$'s adjacency matrix. $G_{\set{H}}$ has $m' = \binom{r}{2} m$ edges and multiplicity $\mu \leq \binom{n-2}{r-2}$ by Lemma~\ref{lem:edge-multiplicity}. By Lemma~\ref{lem:triangle-correspondence}, $G_{\set{H}}$ contains $t$ triangles. Applying Lemma~\ref{lem:matrix-edge-triangle} and Corollary~\ref{cor:multigraph-triangle-bound}, we bound
	\[
		(6t)^2 \leq (2 \mu m')^3.
	\]
	Solving for $t$ gives the desired result.
\end{proof}

In order to complete our proof of Theorem~\ref{thm:clique-lb}, we employ the information theoretic framework of Izumi and Le Gall~\cite{Izumi2017Triangle}. Given the combinatorial bounds proven above, this final step of the argument is the same for simplex and triangle enumeration by substituting the corresponding bounds above. In order to simplify the presentation, we prove the lower bound for only for the triangle enumeration task.

The basic idea is simple: If $H \sim H(n, r, 1/2)$, then for any algorithm $A$ that solves the triangle listing problem, there is some vertex $v$ that outputs at least $\Omega(n^{3r-4})$ triangles with constant probability (by Corollaries~\ref{cor:expected-triangles} and~\ref{cor:triangle-cover}). By Corollary~\ref{cor:triangle-edge-bound}, these triangles must be supported by $\Omega(n^{r-2/3}/r^2)$ edges. That is, $v$ ``witnesses'' the existence of these edges. The information theoretic framework formalizes the intuition that the communication transcript that $v$ receives must therefore contain at least $\Omega(n^{r-2/3}/r^2)$ bits of entropy, as the (non)existence of each edge in $H\sim H(n, r, 1/2)$ has $1$ bit of entropy. The entropy of the transcript is a lower bound for the length of the transcript, hence Theorem~\ref{thm:clique-lb} follows from the fact that $v$ can receive at most $n \log n$ bits each round.

In order to formalize this approach, we introduce some notation and recall some fundamental definitions and results from information theory. Fix the vertex set $V = \set{[n]} = \set{1, 2, \ldots, n}$ and set $m = \binom{n}{r}$. For each set $e \in \binom{[n]}{r}$, we associate a random variable $\bfe_e$ which is $1$ if $e$ is included in $H$ and $0$ otherwise. Let $\calE = (\bfe_1, \bfe_2, \ldots, \bfe_m)$ be the concatenation of the $\bfe_e$.

Fix a triangle enumeration algorithm $\calA$. Let $\rho_v$ be the random variable consisting of the initial state of $v$, which depends only on the edges in $H$ incident to $v$, and let $\rho = (\rho_1, \rho_2, \ldots, \rho_n)$ be the concatenation of these variables. In an execution of $\calA$ and a vertex $v$ let $\pi_v$ and $\calT_v$ denote the transcript of messages received by $v$ and the output, respectively, of $v$, where $\calT_v$ is the set of triangles listed by $v$. We denote $\calT = (\calT_1, \calT_2, \ldots, \calT_n)$ to be the concatenation of the $\calT_v$ so that $\calT$ contains all triangles in $H$. Observe that $\pi_v$ and $\calT_v$ are random variables whose values are determined by $\calE$ and $\rho$. For a set $\calT_v$ of triangles, let $E(\calT_v) \subset E$ denote the set of edges that are included in some triangle in $\calT_v$. That is
\[
	E(\calT_v) = \set{e \in E \sucht \exists T \in \calT_v \text{ with } e \in T}.
\]
We note that in order to show Theorem~\ref{thm:clique-lb}, it is sufficient to show that $\max_v \abs{\pi_v} = \Omega(n^{r - 2/3})$.

We recall that for for a random variable $\bfX$, the \dft{entropy} of $\bfX$ is defined to be $H(\bfX) = - \sum_{x} p(x) \log p(x)$, where the sum is taken over the range of $\bfX$ and $p(x) = \Pr(\bfX = x)$ is the probability density function applied to $x$. Given random variables $\bfX$ and $\bfY$, the \dft{mutual information} between $\bfX$ and $\bfY$ is defined to be $I(\bfX; \bfY) = H(\bfX) - H(\bfX | \bfY)$, where $H(\bfX | \bfY)$ is the conditional entropy of $\bfX$ given $\bfY$. We will use the the following facts about entropy and mutual information.

\begin{lemma}\label{lem:basic-entropy}
	Suppose $\bfX$, $\bfY$, and $\bfZ$ are random variables. Then the following hold:
	\begin{itemize}
		\item $I(\bfX; \bfY) = I(\bfY; \bfX)$,
		\item $I(\bfX; \bfY) \leq H(\bfX)$, and
		\item $I(\bfX; \bfY) \geq I(\bfX; (\bfY, \bfZ)) - I(\bfX; \bfZ)$.
	\end{itemize}
\end{lemma}

\begin{lemma}[Data Processing Inequality]\label{lem:data-processing}
	For any random variables $\bfX$, $\bfY$, and $\bfZ$, if $\bfX$ and $\bfZ$ are conditionally independent given $\bfY$, then $I(\bfX; \bfY) \geq I(\bfX, \bfZ)$.
\end{lemma}

The following lemma, due to Izumi and Le Gall~\cite{Izumi2017Triangle}, formalizes that intuition that the set of triangles $\calT_v$ output by a vertex $v$ contains all of the information about the edges contained in the triangles. The proof in our setting is identical to that of Izumi and Le Gall, so we do not reproduce the argument here.

\begin{lemma}\label{lem:triangle-information}
	For any vertex $v$, we have
	\begin{equation}\label{eqn:triangle-information}
		I(\calE; \calT_v) \geq \E\sqb{E(\abs{\calT_v})}
	\end{equation}
\end{lemma}

We now have all the pieces in place to prove Theorem~\ref{thm:clique-lb}. The proof outline follows that of Izumi and Le Gall~\cite{Izumi2017Triangle}.

\begin{proof}[Proof of Theorem~\ref{thm:clique-lb} (triangle enumeration)]
	Fix any (randomized) algorithm $\calA$ for triangle listing in the CLIQUE model, and consider inputs generated according to $H \sim H(n, r, 1/2)$. Let $\bfw = \arg\max_{v} \set{\abs{\calT_v}}$ denote the (random variable) vertex that outputs the largest number of triangles in an execution of $\calA$. By Corollary~\ref{cor:triangle-cover}, $\abs{\calT_\bfw} = \Omega(n^{3r - 4})$ with constant probability. 

	Applying Lemmas~\ref{lem:basic-entropy} and~\ref{lem:data-processing}, we bound the entropy of the transcript, $\pi_\bfw$, received by $\bfw$ as follows:
	\begin{equation}\label{eqn:pi-entropy}
		\begin{split}
			H(\pi_\bfw) & \geq I(\calE, \pi_\bfw)                                    \\
			            & \geq I(\calE, (\pi_\bfw, \rho_\bfw)) - I(\calE, \rho_\bfw) \\
			            & \geq I(\calE, (\pi_\bfw, \rho_\bfw)) - H(\rho_\bfw)        \\
			            & \geq I(\calE, \calT_\bfw) - H(\rho_\bfw).
		\end{split}
	\end{equation}
	The final inequality uses the data processing inequality (Lemma~\ref{lem:data-processing}) together with the fact that $\calE$ and $\calT_\bfw$ are conditionally independent given $(\pi_\bfw, \rho_\bfw)$ (because the output of $\bfw$ is completely determined by its initial condition and the transcript of messages it receives). The initial state $\rho_\bfw$ contains only the knowledge of (possible) edges incident to $\bfw$, of which there are $\binom{n-1}{r-1}$. Thus we bound
	\begin{equation}\label{eqn:rho-entropy}
		H(\rho_\bfw) = \sum_{e : \bfw \in e} H(\bfe_e) = \binom{n-1}{r-1} = O(n^{r-1}),
	\end{equation}
	where we have used the fact that the $\bfe_e$ are independent Bernoulli random variables with probability $p = 1/2$. Combining inequalities~(\ref{eqn:pi-entropy}) and~(\ref{eqn:rho-entropy}) with the conclusions of Corollaries~\ref{cor:triangle-cover} and~\ref{cor:triangle-edge-bound} and Lemma~\ref{lem:triangle-information}, we obtain
	\begin{equation*}
		H(\pi_\bfw) = \Omega(n^{r-2/3}/r^2) \quad\text{for}\quad r = o(n^{1/6}).
	\end{equation*}
	Since $H(\pi_\bfw)$ is a lower bound for the expected length of $\bfw$'s transcript and $\bfw$ receives $O(n \log n)$ bits each round, the expected number of rounds of $\calA$ must be $\Omega(H(\rho_\bfw)/n \log n) = \Omega(n^{r-5/3}/r^2 \log n)$, as desired.
\end{proof}

The proof for simplex enumeration is identical, except that we use Corollary~\ref{cor:simplex-edge-bound} to bound $\abs{\calT_\bfw} = \Omega(n^{r - 1 + 1/(r+1)})$, and hence $H(\pi_\bfw) = \Omega(n^{r - 1 + 1/(r+1)})$, giving the claimed result.

%% file: general-hypergraphs.tex
\subsection{Triangle Enumeration in General Hypergraphs}
\label{sec:general-hypergraphs}

The matching upper and lower bounds of Sections~\ref{sec:clique-ub} and~\ref{sec:clique-lb} apply to to triangle enumeration in the CLIQUE model. In this section, we consider the triangle enumeration problem in the hypergraph generalizations of the CONGEST model defined in Section~\ref{sec:computational-models}. We focus on the EDGE CLIQUE (EC) model, which is the most powerful of the models defined. Specifically, we show the following theorem, where we use the notation $\widetilde{O}(\cdot)$ to suppress factors that are polylogarithmic in the function's argument.

\begin{theorem}\label{thm:ec-ub}
	Triangle enumeration can be performed in rank $r$ hypergraphs in expected $\widetilde{O}(n^{1/3} \min\set{r \Delta_2, n^{r-2}})$ rounds in the EDGE CLIQUE (EC) and PRIMAL CONGEST (PC) models, where $\Delta_2$ is the pair-degree of $H$.
\end{theorem}

By Proposition~\ref{prop:model-relationships}, the result of Theorem~\ref{thm:ec-ub} holds for all of the models defined in Section~\ref{sec:computational-models} up to factors that are polynomial in $r$.

\begin{remark}
	While the upper bound of Theorem~\ref{thm:ec-ub} matches the (tight) upper bound of Theorem~\ref{thm:clique-ub} for the CLIQUE model, the computational power of CLIQUE and EC are not comparable when $\Delta_2$ is super-constant. To see why, observe that if $u$ and $v$ share $\deg(u, v)$ incident edges in $H$, then in a single round of EC, $u$ and $v$ can exchange $O(\deg(u, v) \log n)$ bits of communication (corresponding to $O(\log n)$ bits per shared edge). Thus, the lower bound of Theorem~\ref{thm:clique-lb} does generally apply to the EC model. On the other hand, a single round of the EC model can be simulated in $\Delta_2$ rounds of the CLIQUE model, hence the upper bound of Theorem~\ref{thm:ec-ub} is tight (up to $\poly r \log n$ factors) when $\Delta_2 = O(1)$.
\end{remark}

The proof Theorem~\ref{thm:ec-ub} applies the techniques developed by Chang et al.~\cite{Chang2021Nearoptimal} for solving (graph) triangle enumeration in the CONGEST model. Chang et al.'s solution develops two main tools to perform triangle enumeration:
\begin{enumerate}
	\item computing an expander decompositions in which the graph $G = (V, E)$ is partitioned into components such that (a)~each component is an expander graph, and (b) only $\eps \abs{E}$ edges are not contained within some component, and
	\item computing efficient message routing within expander graphs.
\end{enumerate}
With these two tools, Chang et al.\ perform triangle enumeration in the CONGEST model by first computing an expander decomposition of the graph. Then within each component of the decomposition, the efficient message routing procedure is applied to simulate the CLIQUE triangle enumeration algorithm of Dolev et al.~\cite{Dolev2012Tri} within each component (with only $\poly\log n$ overhead) to enumerate all triangles with at least two vertices in the component. This two step procedure is then applied recursively on the subgraph of inter-cluster edges (of which there are at most $\eps \abs{E}$).

Our approach is to simulate both steps of Chang et al.~in the EC model. There are two main observations that allow us to perform this simulation. First, we observe that if $G = G(H)$ is the primal graph of $H$, then a single round of CONGEST in $G$ can be simulated by a single round of computation in the EC model. Thus, any CONGEST algorithm on $G$ can be simulated in $H$ in the EC model without additional computational overhead. Second, we note that triangle enumeration in a hypergraph $H$ is equivalent to triangle enumeration in the multigraph $G_{[H]}$ in which the multiplicity of each pair $\set{u, v}$ is equal to the number of edges $e \in E$ such that $u, v \in e$. Using these two observations we perform triangle enumeration in $H$ as follows:
\begin{enumerate}
	\item Compute an expander decomposition of $G = G(H)$ with resulting components $G_1, G_2, \ldots$.
	\item Within each component $G_i$ perform triangle enumeration, listing triangles of the form $T = v_1 e_1 v_2 e_2 v_3 e_3$  such that at least two of the $v_i$ are in $G_i$ as follows:
	      \begin{enumerate}
		      \item simulate computation of the message the routing of~\cite{Chang2021Nearoptimal} for triangle enumeration in $G_i$,
		      \item if $u$ sends the edge $e' = \set{u, v}$ to a vertex $w$ in the simulation, then $u$ sends to $w$ all of its incident hyperedges $e$ with $u, v \in e$.
	      \end{enumerate}
	\item Recursively enumerate triangles containing only inter-cluster edges in $G$.
\end{enumerate}

The simulation of Steps~1, 2b, and~3 are straightforward given the our observation that CONGEST computation in $G(H)$ can be simulated in $H$ in the EC model. To perform Step~2b, we note that the size of the message sent from $u$ to $w$ can be made to be $O(\min\set{r \Delta_2, n^{r-2} + \log n})$ where $r$ is the rank of $H$. The first term can be achieved by each vertex sending the list of neighboring vertices in each incident edge, while the second is achieved by sending the identities of $u$ and $v$ followed by a binary vector of length $n^{r-2}$ where the indices enumerate all of the possible $O(n^{r-2})$ edges that contain $u$ and $v$. Clearly, these longer messages can be routed to their respective destinations with a multiplicative overhead of $O(\min \set{r \Delta, n^{r-2}/\log n + 1})$ rounds compared to the algorithm of~\cite{Chang2021Nearoptimal}.

%% file: sparse-hypergraphs.tex
\section{Triangle and Simplex Enumeration in Sparse Hypergraphs}\label{sec:sparse}

Our results in Section~\ref{sec:clique} provide existentially optimal algorithms for triangle and simplex enumeration in the CLIQUE and PRIMAL CONGEST (PC) models parameterized by the rank of a hypergraph. The upper bound may have limited practical application, however, because the bound is exponential in the rank $r$ of the hypergraph. Many natural hypergraphs (such as those representing social networks, for example) may have very large rank, making this algorithm unsuitable. On the other hand, the lower bound construction in Section~\ref{sec:clique} contains a number of edges that is exponential in $r$. Again this is a feature that many naturally occurring hypergraphs do not possess.

In this section, we consider families of sparse hypergraphs that generalize sparse graphs. We define two notions of sparseness: global sparseness and everywhere sparseness. A (globally) sparse hypergraph is a hypergraph in which the sum of the degrees of vertices (equivalently the sum of the cardinalities of edges) is linear in the number of vertices. This definition is a direct generalization of a sparse graph. An \emph{everywhere sparse} hypergraph is a hypergraph in which every sub-hypergraph is (globally) sparse. We give a formal definition of this notion in Section~\ref{sec:everywhere-sparse}. Everywhere sparseness generalizes the asymptotically equivalent notions of degeneracy and arboricity of a graph. Just as many naturally occurring graphs have bounded degeneracy (and arboricity), we expect that many natural families of hypergraphs are everywhere sparse.

In this section we give upper bounds for triangle and simplex enumeration in sparse and everywhere sparse hypergraphs. These algorithms can be implemented in the EDGE BROADCAST (EB) model. We complement our upper bounds with lower bounds as well, though these bounds are not tight. We leave it as an open question whether the algorithms or lower bounds can be improved.

\subsection{Low Degree Triangles}

Let $H = (V, E)$ be a hypergraph of rank $r$. We say that a vertex $v \in V$ is \dft{$\delta$-light} if $\deg(v) \leq \delta$. We call a triangle $T = (v_1, e_1, v_2, e_2, v_3, e_3)$ $\delta$-light if at least one of $v_1, v_2$ and $v_3$ is $\delta$-light. Similarly, an $r$-dimensional simplex $S = \set{v_0, v_1, \ldots, v_r}$ is $\delta$-light if at least one of the $v_i$ is $\delta$-light.  The following result shows that we can efficiently enumerate all $\delta$-light triangles and simplices in $H$ that contain at least one light vertex.

\begin{proposition}\label{prop:light-triangle}
	Suppose $H = (V, E)$ is a hypergraph of rank $r$, and let $\delta > 0$ be a parameter. Then all $\delta$-light triangles and simplices in $H$ can be listed in $O(r \delta)$ rounds of the EDGE BROADCAST (EB) and PRIMAL CONGEST (PC) models. Moreover, each $\delta$-light triangle is listed by exactly one vertex in $H$.
\end{proposition}
\begin{proof}
	Consider the following procedure: each $\delta$-light vertex $v$ broadcasts its $1$-neighborhood to all of its neighbors. Specifically, for each incident edge $e$ to $v$, $v$ broadcasts $e$'s ID together with the list of vertices contained in $e$. Observe that this broadcast requires $O(r \delta)$ rounds in the EB and PC models, as each $v$ sends  $\deg(v) \leq \delta$ lists each containing (at most) $r$ IDs to its neighbors.

	We first consider the case of triangle enumeration. Suppose $w$ is a neighbor of $v$ where $v$ and $w$ are both incident to edge $e$. We claim that after $v$'s broadcast, $w$ knows every triangle $T$ containing the $2$-path $(v, e, w)$. To see this, observe that $w$ receives the pair $u, f$ for every edge $f$ incident to $v$ and vertex $u \in f$. Further, $w$ knows every edge $g$ incident to itself, in particular, every $g$ containing a given vertex $u$. Thus, $w$ knows every triangle $(u, f, v, e, w, g)$.

	To ensure that each $\delta$-light triangle $T$ is listed by exactly one vertex, we apply the following rule:
	\begin{itemize}
		\item If $T$ contains at least one vertex that is $\delta$-heavy, then the $\delta$-heavy vertex with the smaller ID outputs $T$.
		\item If $T$ contains only $\delta$-light vertices, then the vertex with the smallest ID outputs $T$.
	\end{itemize}
	Note that each vertex can determine for itself if it should output $T$ because it knows both the IDs of all its neighbors as well as their degrees.

	For case of simplex enumeration, suppose $S = \set{v_0, v_1, \ldots, v_r}$ is an $(r+1)$-subset of $V$ and $v_0$ is $\delta$-light. Observe that $S$ is a simplex if and only if every $r$-subset in $\binom{S}{r}$ is an edge. Each vertex $v_i$ is incident to all of these edges (if present in $H$) except $e_{-i} = S \setminus \set{v_i}$. Thus, in order for $v_i$ to determine if $S$ is a simplex, $v_i$ need only learn about the existence of $e_{-i}$. Since $v_0 \in e_{-i}$ for all $i \neq 0$, if $v_0$ sends its incident edges to $v_i$ for all $i \neq 0$, each of the $v_i$ can determine whether $S$ is a simplex individually without further communication. In order to ensure that each simplex is output by exactly one vertex, we can use the same tie-breaking rule as the triangle case, choosing the $\delta$-heavy vertex (if any) with the smallest ID, or the the vertex with the smallest ID if all $v_i$ are $\delta$-light. In either case, every $\delta$-light simplex is output by exactly one vertex, as required.
\end{proof}

\subsection{Sparse hypergraphs}

A hypergraph $H = (V, E)$ is \dft{sparse} if $\sum_{e \in E} \abs{e} = O(n)$, where $n = \abs{V}$. Observe that $\sum_{e \in E} \abs{e} = \sum_{v \in V} \deg(d)$. Also, if $H$ is sparse, the maximal degree $\Delta$ of a vertex is $O(n)$. In this section, we will show that triangle and simplex enumeration in sparse hypergraphs can be performed in $O(n)$ rounds of the EB and PC models. We show this upper bound is essentially tight for the PC model, as triangle and simplex enumeration require $\Omega(n^{1-\eps})$ rounds in sparse hypergraphs for every $\eps > 0$ in PC. The upper bound follows from the following more general result for bounded-degree hypergraphs.

\begin{theorem}\label{thm:bounded-degree-ub}
	Suppose $H = (V, E)$ has maximum degree $\Delta$. Then triangle and simplex enumeration can be performed in $O(\Delta)$ rounds of the EDGE BROADCAST (EB) and PRIMAL CONGEST (PC) models.
\end{theorem}
\begin{proof}
	We assume that each edge $e$ has a unique $O(\log n)$-bit ID. In the EB model with the $\KT_1$ assumption, such an assignment of edge IDs can be performed in a single round (see Section~\ref{sec:edge-ids}). Then in $O(\Delta)$ rounds, each vertex broadcasts the IDs of all its incident edges to all of its neighbors.

	We first consider the case of triangle enumeration. Let $v_1$ be a vertex with with neighboring vertices $v_0$ and $v_2$, where $v_0, v_1 \in e_0$ and $v_1, v_2 \in e_1$. Note that $T = (v_0, e_0, v_1, e_1, v_2, e_2)$ form a (closed) triangle for some edge $e_2$ if and only if the following conditions hold:
	\begin{enumerate}
		\item $v_0, v_2 \in e_2$
		\item $v_1 \notin e_2$
		\item $v_2 \notin e_0$
		\item $v_0 \notin e_1$
	\end{enumerate}
	Note that Condition~1 holds if and only if $v_1$ receives the message $e_2$ from both $v_0$ and $v_1$ when the latter vertices broadcast their incident edges. Further, $v_1$ can verify Conditions~2--4 locally, as $v_1$ knows the IDs of the vertices in its incident edges ($e_1$ and $e_0$), and $v_1$ knows whether or not $v_1 \in e_2$. Thus, after the broadcast, $v_1$ can determine every incident triangle.

	By symmetry, $v_0$ and $v_2$ also witness the triangle $T$, hence if each triangle should only be output by a single vertex, we can require that only the incident vertex with smallest ID outputs $T$.

	Now consider the case of simplex enumeration. As in the triangle case, each vertex broadcasts the IDs of its incident edges to its neighbors in $O(\Delta)$ rounds of EB or PC. Observe that $S = \set{v_0, v_1, \ldots, v_r}$ forms a simplex if and only if for every index $i$, we have $e_i \in E$ where $e_i = S \setminus \set{v_i}$. Since $v_i$ is incident to every (potential) edge $e_j$ for $j \neq i$, in order to determine if $S$ is a simplex, $v_i$ need only learn whether or not $e_i \in E$. While $v_i$ does not initially know the ID of $e_i$ (if present), $v_i$ can determine if $e_i$ is present from its received messages because $e_i$ is the unique edge for which $v_i$ would receive the same ID from all $v_j$, $j \neq i$. Thus, $v_i$ determines that $e_i$ exists (and $S$ is a simplex) if and only if (1)~every edge $e_j$, $j \neq i$ exists and~(2) there is an edge ID that $v_i$ receives from all $v_j$, $j \neq i$. By symmetry, all vertices $v_i$ learn of the existence of the simplex $S$, hence the $v_i$ with the smallest ID can output $S$, if present.
\end{proof}

The following corollary follows immediately from the fact that sparse hypergraphs have maximum degree satisfying $\Delta = O(n)$.

\begin{corollary}\label{cor:sparse-ub}
	Suppose $H = (V, E)$ is a sparse hypergraph. Then triangle and simplex enumeration can be performed in $O(n)$ rounds on $H$ in the EB and PC models.
\end{corollary}

We now prove a nearly matching lower bound for sparse hypergraphs for PC model.

\begin{theorem}\label{thm:sparse-lb}
	For every $\eps > 0$, triangle and simplex enumeration in sparse hypergraphs requires $\Omega(n^{1 - \eps})$ rounds in the PRIMAL CONGEST (PC) model.
\end{theorem}
\begin{proof}
	The idea of the proof is to construct a distribution of sparse hypergraphs that contain smaller dense sub-hypergraphs to which we can apply our general CLIQUE lower bound (Theorem~\ref{thm:clique-lb}). More concretely, let $r = \lceil{2 / \eps}\rceil$, and for any $n$, let $n' = \lceil{n^{1/r}}\rceil$. We form a distribution of hypergraphs $H$ on the vertex set $V = \set{1, 2, \ldots, n}$ as follows.
	\begin{itemize}
		\item For vertices in $[n'] = \set{1, 2, \ldots, n}$ and each $r$-subset $e$ of $[n']$, include $e \in E$ independently with probability $1/2$.
		\item For each $v \leq n'$, add a rank-$2$ hyperedge $\set{v, n'+1}$ to $E$.
		\item Add a single hyperedge $\set{n'+1, n'+2, \ldots, n}$ to $E$.
	\end{itemize}
	Note that the edges in Item~1 are distributed according to $G(n', r, 1/2)$. The remaining edges simply ensure that the hypergraph is connected while maintaining global sparseness.

	First, we give an upper bound on the number and size of edges in $H$. In Step~1 above, at most $\binom{n'}{r} = O((n')^r)$ edges are added, and each has size $r$. Step~2 adds another $n' = O(n^{1/r})$ edges, each of rank~2. Finally, Step~3 adds a single edge of size $n - n' = O(n)$. Summing the sizes of the edges, we obtain
	\begin{equation*}
		\sum_{e \in E} \abs{e} \leq r (n')^r + n' + (n - n') = r (n^{1/r})^r + n = O(r n).
	\end{equation*}
	Thus, for constant $r = O(1/\eps)$, $H$ is sparse.

	Next, we argue that each round of PC can be simulated by $O(1)$ rounds of the CLIQUE model in $H' = H\left[[n']\right]$, the restriction of $H$ to vertices in $[n']$. Clearly, a single round of EC in $H'$ can be simulated by a single round of CLIQUE on $H'$. In an additional round of CLIQUE, each $v \in [n']$ can broadcast the message it sent to the vertex $n' + 1$ so that each $v$ learns all messages received by $n'+1$ in the round. Arguing by induction on the rounds, each $v \in [n']$ can simulate the states of all vertices $w \in [n] \setminus [n']$ in the simulated EC round, as the initial states and edges incident to these vertices are known to all vertices. Thus, 2 rounds of CLIQUE on $H'$ are sufficient to simulate a single round of EC on $H$.

	For the triangle enumeration lower bound, we next apply the lower bound of Theorem~\ref{thm:clique-lb}, which implies that $\Omega((n')^{r-5/3})$ rounds are necessary to perform triangle enumeration in $H'$ in CLIQUE. The simulation result above implies that if triangle enumeration could be performed in $t$ rounds of EC in $H$, then triangle enumeration could be performed in $2t$ rounds of CLIQUE. Thus, $\Omega((n')^{r-5/3})$ rounds are necessary for triangle enumeration in EC on graphs constructed as above. The theorem then follows by substituting $n' = n^{1/r}$:
	\begin{equation*}
		(n')^{r-5/3} = (n^{1/r})^{r - 5/3} = n^{1 - 5/3r} \geq n^{1-\eps}.
	\end{equation*}

	The argument for the simplex enumeration is analogous. Theorem~\ref{thm:s-clique-ub} for instances of size $n' = n^{1/r}$ where $r$ is sufficiently large that $2/r - 1/(r^2 + r) \leq \varepsilon$ then gives a lower bound of $\Omega(n^{1-\varepsilon})$.
\end{proof}

While Theorem~\ref{thm:sparse-lb} shows that the simple algorithms of Theorem~\ref{thm:bounded-degree-ub} are essentially optimal for sparse hypergraphs in the PC model, the same lower bound is not implied for the EB model. Indeed, the construction in the lower bound argument has maximum 2-degree $\Delta_2 = \Theta(n^{1-2/r})$ with high probability. Applying Proposition~\ref{prop:pc-relationship}, the implied lower bound for the EB model is only $\Omega(n^{1/3r})$. Taking $r = 2$, we obtain the following corollary.

\begin{corollary}\label{cor:sparse-eb-lb}
	Triangle and simplex enumeration in sparse hypergraphs require $\Omega(n^{1/6})$ rounds in the EC, EB, EU, ES, and EP models.
\end{corollary}

While this bound is far from tight, it shows that a polynomial number of rounds is required for triangle enumeration in sparse hypergraphs, even in the the strong EDGE CLIQUE (EC) model.


\subsection{Everywhere Sparse Hypergraphs}
\label{sec:everywhere-sparse}

The $\Omega(n^{1-\eps})$ lower bounds of Theorem~\ref{thm:sparse-lb} crucially employ hypergraphs with small but dense subsets of vertices (namely, the set $[n']$ in the construction). Thus, the hypergraphs are globally sparse, having average degree $O(1)$, but locally dense, as the average degree in $H' = H\left[n'\right]$ is $\Theta(n^{1-1/r})$ in expectation. On the other hand, the upper bound of Theorem~\ref{thm:bounded-degree-ub} implies that if \emph{all} degrees are bounded, then triangle enumeration can be performed efficiently. In this section, we consider more refined hypergraph parameter that measures the maximum density of sub-hypergraphs of $H$.

In the case of \emph{graphs}, the idea of ``maximum density subgraph'' is captured by the graph's arboricity. Here the \dft{arboricity}, $\alpha(G)$, of a graph $G = (V, E)$ is defined to be the minimum number of forests into which $G$'s edges can be partitioned. The seminal result of Nash-Williams shows that this definition is equivalent to all subgraphs of $G$ being sparse.

\begin{theorem}[Nash-Williams~\cite{Nash-Williams1964Decomposition}]\label{thm:nash-williams}
	Suppose $G = (V, E)$ is a graph. Then
	\begin{equation}
		\alpha(G) = \max_{U \subseteq V} \frac{\abs{E[U]}}{\abs{U} - 1},
	\end{equation}
	where $E[U]$ denotes the set of induced edges of $G$ within the set $U$.
\end{theorem}

For hypergraphs, we define an analogous measure of maximum sub-hypergraph density. Before introducing this measure, we define the restriction of a hypergraph $H = (V, E)$ to a subset $U$ of vertices.

\begin{definition}\label{dfn:restriction}
	Let $H = (V, E)$ be a hypergraph and $U \subseteq V$ a non-empty subset of vertices. We define the \dft{restriction} of $H$ to $U$ to be the hypergraph $H[U] = (U, E[U])$ where
	\begin{equation*}
		E[U] = \set{e \cap U \sucht e \in E, \abs{e \cap U} \geq 2}.
	\end{equation*}
\end{definition}

\begin{definition}\label{dfn:max-density}
	Let $H = (V, E)$ be a hypergraph and $U \subseteq V$ a non-empty subset. We define the \dft{density} of $U$ relative to $H$, denoted $\mu_H(U)$ to be
	\begin{equation}\label{eqn:density}
		\mu_H(U) = \frac{1}{\abs{U}} \sum_{e \in E[U]} \abs{e}.
	\end{equation}
	We define the \dft{maximum density} of $H$, $\mu(H)$, to be
	\begin{equation}
		\mu(H) = \max_{U \subseteq V} \mu_H(U),
	\end{equation}
	where the maximum is taken over all non-empty subsets of $V$.
\end{definition}

The maximum density is analogous to the arboricity in the case of graphs, as if $H$ is a graph (i.e., rank-$2$ hypergraph), then $\alpha(H) \leq \mu(H) \leq 2 \alpha(H)$. 
In distributed graph algorithms, arboricity has been used to parameterize the running time of graph algorithms. This technique, developed by Barenboim and Elkin~\cite{Barenboim2013Distributed}, relies on a structural result for graphs of bounded arboricity. Namely, the vertices graphs with arboricity $\alpha$ can be partitioned into $O(\log n)$ layers, where each layer has $O(\alpha)$ neighbors in higher layers. This structure can then be used to perform a distributed computation on the layers sequentially, where each layer can be treated has having maximum degree $O(\alpha)$. We show that we can use maximum density in essentially the same way.

\begin{proposition}\label{prop:max-density}
	Suppose $H = (V, E)$ is a hypergraph with maximum density $\mu(H) \leq \alpha$, for some constant $\alpha$. Then the vertices $V$ of $H$ can be partitioned into $\ell = \lceil \log n \rceil$ layers, $V = L_1 \cup L_2 \cup \cdots \cup L_\ell$ such that each vertex $v \in L_i$ has degree at most $2 \alpha$ in $H_i = H[L_i \cup L_{i+1} \cup \cdots \cup L_\ell]$.
\end{proposition}
\begin{proof}
	Define $V_0 = V$, $L_1 = \set{v \in V \sucht \deg(v) \leq 2 \alpha}$, and take $V_1 = V \setminus L_1$. Given $L_i$ and $V_i$, inductively define $H_{i+1} = H[V_i]$, $L_{i+1} = \set{v \in V_i \sucht \deg_{H_{i+1}} \leq 2 \alpha}$, and $V_{i+1} = V_i \setminus L_i$. We will show that for $\ell = \lceil \log n \rceil$, the sequence $L_1, L_2, \ldots, L_\ell$ is a partition of the vertices as above.

	We claim that for each $i$, $\abs{L_i} \geq \frac{1}{2} \abs{V_{i-1}}$, hence $\abs{V_{i+1}} \leq \frac{1}{2} \abs{V_i}$. By induction, this implies that $\abs{V_i} \leq 2^{-i} \abs{V_0} = n / 2^i$, whence the proposition follows.

	To prove the claim, suppose to the contrary that $\abs{L_i} < \frac{1}{2} \abs{V_{i-1}}$. That is, more than half of the vertices $v \in V_{i-1}$ have degree larger than $2 \alpha$. Then we compute:
	\begin{align*}
		\mu_H(V_{i-1}) & = \sum_{e \in E[V_{i-1}]} \abs{e}          \\
		               & = \sum_{v \in V_{i-1}} \deg_{H_{i-1}} (v)  \\
		               & \geq \sum_{v \in V_{i}} \deg_{H_{i-1}} (v) \\
		               & > \frac 1 2 \abs{V_{i-1}} (2 \alpha)       \\
		               & = \alpha \abs{V_{i-1}}.
	\end{align*}
	The first inequality holds because $V_{i} \subseteq V_{i-1}$ by construction. The second inequality holds by the assumption that $\abs{V_{i}} > \frac 1 2 \abs{V_{i-1}}$ and definition of $V_{i}$. This contradicts the assumption that $\mu(H) < \alpha$, so the desired result follows.
\end{proof}

Proposition~\ref{prop:max-density} suggests the following strategy for triangle enumeration in hypergraphs $H$ with bounded $\mu(H)$: given the sequence $L_1, L_2, \ldots, L_\ell$ of layers, process the layers sequentially. When vertices in layer $i$ become active, they broadcast their (at most $2 \mu(H)$) incident edges (complete neighbor list) that intersect $L_i \cup \cdots \cup L_\ell$ to all neighbors in $O(\mu(H) \rank(H))$ rounds. After all vertices in $L_1 \cup \cdots \cup L_i$ are processed, each vertex $v \in L_i$ will have received the incident edges and neighboring vertices of all its neighbors in $L_1 \cup \cdots \cup L_i$. Thus, $v$ can determine all triangles $T = (v, e_1, u, e_2, w, e_3)$ containing a vertex $u \in L_1 \cup \cdots \cup L_i$. If $u$ and $w$ were in the same layer $L_j$ with $j < i$, then $v$ is responsible for outputting $T$. If $v$ and $w$ are both in layers strictly after $L_j$, or if $u$, $v$, and $w$ are in the same layer, then whichever vertex has the smaller ID is responsible for outputting $T$.


If the vertices are given an upper bound, $\alpha$, on $\mu(H)$ and each vertex $v$ is given the index $i$ for which $v \in L_i$, then the procedure above can easily be simulated in $O(\alpha r \log n)$ rounds of EDGE BROADCAST or PRIMAL CONGEST if $H$ has rank $r$. Even if the vertices do not know their layer indices but do know $\alpha$, they can still simulate the procedure as follows. Consider phases lasting $O(\alpha r)$ rounds. Each phase, a vertex can be waiting, broadcasting, or completed. A vertex is broadcasting in a phase $i$ if it has at most $2 \alpha$ incident edges containing vertices that were waiting in the previous phase. When a vertex is broadcasting, it broadcasts its incident edges in $O(\alpha r)$ rounds. In subsequent rounds, the vertex is in the completed state and performs no further computation. We next show that prior knowledge of an upper bound $\alpha$ on $\mu$  is not necessary in order to implement this procedure.

\subsubsection{Triangle Enumeration Without Prior Knowledge of \texorpdfstring{$\mu$}{mu}}


Here we consider the case where vertices $v$ do not any prior knowledge of $\mu(H)$. For any $\alpha > 0$, consider the following method for testing whether the maximum density of $H$ is at most $2 \alpha$. Each round each vertex is either active or inactive, with all vertices being initialized to the active state. In every round, each active $v$ examines the state of its neighbors. If $v$ has at most $2 \alpha$ active neighbors, then $v$ sets its state to inactive. Otherwise $v$ remains active. The proof of Proposition~\ref{prop:max-density} implies that if $H$ has maximum density at most $\alpha$, then every vertex will be inactive after at most $\lceil \log n \rceil$ rounds. We formalize this procedure with the function $\peel$ defined in Algorithm~\ref{alg:peel-routines}.

\begin{algorithm}
	\caption{The peel algorithm for obtaining local estimates of the maximum density of a hypergraph.\label{alg:peel-routines}}
	\begin{algorithmic}[1]
		\Function{Peel}{$\alpha$} \Comment{$\alpha$ is a guess for the maximum density of $H$}
		\State set state to ``active''
		\For{round $i = 1, 2, \ldots, \lceil \log n \rceil$}
		\State broadcast state to neighbors
		\If{$v$ has $\leq 2 \alpha$ incident edges containing active vertices $\neq v$}
		\State set state to inactive
		\EndIf
		\EndFor
		\State \Return state
		\EndFunction
		\Function{ParallelPeel}{$M$} \Comment{$M$ an a priori upper bound on maximum density}
		\For{$i = 1, 2, \ldots, \lceil\log M\rceil$ \textbf{in parallel}}
		\State Simulate $\peel(2^i)$
		\EndFor
		\State $\alpha_v \gets 2^i$ where $i$ is the minimum index for which $v$ was inactive at the termination of $\peel(2^i)$
		\State \Return $\alpha_v$
		\EndFunction
		\Function{Flood}{$\alpha_v$}
		\State $\beta_v \gets \alpha_v$
		\For{$i = 1, 2, \ldots, \lceil\log n\rceil$}
		\State broadcast $\beta_v$ to all neighbors
		\State $m \gets$ maximum value received
		\If{$m > \beta_v$}
		\State $\beta_v \gets m$
		\EndIf
		\EndFor
		\State \Return $\beta_v$
		\EndFunction
	\end{algorithmic}
\end{algorithm}

\begin{lemma}\label{lem:peel}
	Suppose $H$ has maximum density $\mu$. Then at the end of an execution of $\peel(\mu)$, every vertex is in the inactive state. Conversely, if an execution of $\peel(\alpha)$ results some vertex being active when the procedure terminates, then $\mu > \alpha/2$.
\end{lemma}
\begin{proof}
	The first claim of the lemma follows immediately from the proof of Proposition~\ref{prop:max-density}.

	Suppose there is some vertex $v$ that does not become inactive during the execution. For $i = 0, 1, \ldots, \lceil \log n \rceil$, let $A_i$ denote the set of vertices active after round $i$ (where $A_0 = V$). Since $A_{\lceil \log n \rceil} \neq \varnothing$, it must be the case that there is an index $j$ for which $\abs{A_{j+1}} > \frac 1 2 \abs{A_j}$. Thus, at least half of the vertices in $A_j$ have degree more than $2 \alpha$ in $H[A_j]$. Therefore,
	\begin{equation*}
		\sum_{v \in A_j} \deg_{A_j}(v) > \frac 1 2 \abs{A_j} (2 \alpha) = \alpha \abs{A_j},
	\end{equation*}
	which implies that $\mu > \alpha$, as desired.
\end{proof}

Lemma~\ref{lem:peel} suggests the following procedure for estimating the maximum density of a hypergraph $H$: run $\peel(\alpha)$ for geometrically increasing values of $\alpha = 2, 4, 8, \ldots$ up to some \emph{a priori} upper bound $M$ on $\mu(H)$. The smallest such value of $\alpha$ for which all vertices terminate in an inactive state is a $4$-approximation of $\mu(H)$. Since $\peel$ only uses $1$ bit of communication per edge per round, we can in fact run all of these instantiations of $\peel$ in parallel using $\log M$ bits of communication per edge per round. We call this procedure $\ppeel$. Assuming $M$ is polynomial in $n$, $\ppeel$ terminates in $O(\log n)$ rounds of the EB or PC models.

Unfortunately, running $\ppeel$ as a distributed procedure does not give each vertex an estimate of the global maximum density. Instead, the value $\alpha_v$ is only a local estimate of the maximum density. In the function $\flood$, each vertex adopts a value $\beta_v$ that is the maximum estimate of $\mu(H)$ of any vertex in its $\lceil\log n\rceil$-hop neighborhood. Since $\max_v \alpha_v \leq 4 \mu(H)$, we have $\beta_v \leq 4 \mu(H)$. Moreover, we show that if each vertex $v$ runs $\peel$ using $\beta_v$ as its estimate, then all vertices will terminate after $\lceil\log n\rceil$ rounds.

\begin{lemma}\label{lem:local-density}
	Suppose the $\peel$ procedure is run where each vertex $v$ uses the input value $\alpha \gets \beta_v$. Then when $\peel$ terminates, all vertices are in the inactive state.
\end{lemma}
\begin{proof}
	Consider a vertex $v \in V$. Let $\alpha_v$ be the value returned by $\ppeel$ at vertex $v$, and let $\beta_v$ be the value returned by the subsequent call to $\flood(\alpha_v)$. Consider an execution of $\peel(\beta_v)$. Here $B_k(v)$ denotes the $k$-hop neighborhood of $v$ in $H$.

	We claim that for all $i = 1, 2, \ldots, \lceil \log n \rceil$, if $u \in B_{\lceil \log n \rceil - i}(v)$ and $u$ became inactive in round $i$ during the (uniform) execution\footnote{By ``uniform execution'', we mean the execution of $\peel$ in which all vertices used $\alpha_v$ as their local input. This is in contrast to the ``nonuniform execution'' in which each vertex $w$ uses $\beta_w$ as its local input.} of $\peel(\alpha_v)$, then $u$ became inactive by the end round $i$ in the (nonuniform) execution of $\peel(\beta_v)$.

	We prove the claim by induction on $i$. For the base case, suppose $u$ became inactive in the first round of the uniform execution. Therefore, we have $\deg(u) \leq \alpha_v$. Since $u \in B_{\lceil \log n \rceil}(v)$, we have $\beta_u \geq \alpha_v$, hence $\deg(u) \leq \beta_u$, so $u$ became inactive after the first round of the nonuniform execution as well.

	For the inductive step, suppose the conclusion of the claim holds after round $i$. Consider round $i + 1$, and let $u \in B_{\lceil\log n\rceil - i - 1}(v)$. By the inductive hypothesis, all vertices in $B_{\lceil\log n\rceil - i}(v)$ that became inactive in round $i$ of the uniform execution also became inactive in round $i$ of the nonuniform execution. Since $B_{\lceil\log n\rceil - i}(v)$ includes all neighbors of $u$, this means that all neighboring vertices of $u$ that became inactive in the uniform execution became inactive in the nonuniform execution. Thus, the degree of $u$ in the nonuniform execution is at most its degree in the uniform execution. Since $\beta_u \geq \alpha_v$ and $u$'s degree in the uniform execution was at most $2 \alpha_v$, $u$'s degree is at most $2 \beta_u$ in the nonuniform execution, hence $u$ becomes inactive in this round. The claim follows.

	Finally, applying the claim for $i = \lceil\log n\rceil$, we find that $v$ becomes inactive by the final round of the execution, as desired.
\end{proof}

The proof of the following lemma is straightforward.

\begin{lemma}\label{lem:ppeel-running-time}
	The procedure $\ppeel(M)$ has running time
	\begin{equation*}
		O(\log n \cdot \max\set{1, \frac{\log M}{\log n}})
	\end{equation*}
	in the Edge Broadcast model.
\end{lemma}

Regarding the parameter $M$ in the $\ppeel$ function, we note that it is sufficient to take $M$ to be (an upper bound on) the maximum degree of a vertex in $H$. Even more pessimistically, for hypergraphs of rank $r$, we can even take $M = \binom{n}{r}$, which results in a running time of $O(r + \log n)$. Since this expression is dominated by the running time of the triangle enumeration algorithm that follows, it is sufficient to use this crude bound when applying $\ppeel$ as a subroutine.

The idea of our triangle and simplex enumeration algorithm is simple. We run $\ppeel$ then $\flood$ so that each vertex $v$ obtains its local maximum density estimate, $\beta_v$. The network then runs $\peel$, where each vertex uses $\beta_v$ as its local input. The vertex $v$ then records the round $i$ in which it becomes inactive, as well as its incident edges containing active vertex (other than $v$ itself) in round $i$. Finally, $v$ broadcasts all of these edges to its neighborhood.

\begin{algorithm}
	\caption{Triangle and simplex enumeration algorithm for bounded maximum density.}
	\begin{algorithmic}[1]
		\Function{EnumerateMotifs}{$n, r$}
		\State $\alpha_v \gets \ppeel\paren{\binom{n}{r}}$
		\State $\beta_v \gets \flood(\alpha_v)$
		\State execute $\peel(\beta_v)$\label{ln:peel}
		\State $t_v \gets$ round in which $v$ became inactive in the execution of $\peel(\beta_v)$
		\State $E_v \gets$ set of edges incident to $v$ containing active vertices (other than $v$) in round $t_v$\label{ln:active-neighbors}
		\State broadcast $E_v$ to all neighbors\label{ln:edge-broadcast}
		\State \Return set of triangles or simplices that contain $v$ contained in the union of received edges and edges incident to $v$
		\EndFunction
	\end{algorithmic}
\end{algorithm}

\begin{remark}
	As stated, triangles or simplices output by vertices in $\enumeratetriangles$ will generally be output by multiple vertices (contained in the triangle/simplex). We can ensure that each motif is only output by a single vertex by modifying the procedure slightly. Specifically, we add a round in which every vertex $v$ broadcasts $t_v$ to its neighbors. In the case of triangles, $v$ outputs incident triangles $T = (v, e_1, u, e_2, w, e_3)$ for which $t_v = \max\set{t_v, t_u, t_w}$ and has minimal ID in the case where the maximum is not unique. Similarly, for simplices $S = \set{v, v_1, v_2, \ldots, v_r}$, $v$ will only output $S$ if $t_v$ was the latest termination time among the $v_i$, breaking ties by ID.
\end{remark}

\begin{theorem}\label{thm:maximum-density-enumeration}
	The function $\enumeratetriangles$ solves the triangle and simplex enumeration problems in $O(\mu(H) r + \log n)$ rounds of the EDGE BROADCAST (EB) and PRIMAL CONGEST (PC) models in hypergraphs with $n$ vertices, rank $r$, and maximum density $\mu(H)$.
\end{theorem}
\begin{proof}
	For the correctness of the procedure, first consider the case of triangles. Let $T = (v_1, e_1, v_2, e_2, v_3, e_3)$ be a triangle. Suppose $v_1, v_2$ and $v_3$ become inactive in rounds $t_1$, $t_2$, and $t_3$ (respectively) in the execution of $\peel$ in Line~\ref{ln:peel} of $\enumeratetriangles$. Assume without loss of generality that $t_1 \leq t_2 \leq t_3$. Since $t_1 \leq t_2, t_3$, $v_2$ and $v_3$ were active in round $t_1$, we have $e_1, e_3 \in E_{v_1}$ in Line~\ref{ln:active-neighbors}. Therefore, these edges are broadcast by $v_1$. In particular, $v_2$ receives the contents of $e_3$ and $v_3$ receives the contents of $e_1$. Thus, both $v_2$ and $v_3$ infer the existence of $T$.

	For the case of simplices, suppose $S = \set{v_0, v_1, \ldots, v_r}$ is a simplex in $H$. Suppose each $v_i$ becomes inactive in round $t_i$ for $i = 0, 1, \ldots, r$ with $t_0 \leq t_1 \leq \cdots \leq t_r$. Then in round $t_0$, $v_0$ broadcasts its incident edges, in particular, $e_1, e_2, \ldots, e_r$ where $e_i = S \setminus \set{v_i}$. Thus, each $v_i$ learns $e_i$ for $i \geq 0$, from which $v_i$ can infer that $S$ is a simplex in $H$. On the other hand, if $S$ is not a simplex, then clearly no vertex $v_i$ will output $S$ because either $v_i$ will witness that some $e_j \notin E$ for $j \neq i$ (as $v_i$ would be incident to $e_j$) or $v_i$ will not receive $e_i$.

	To analyze the running time, we first observe that the running time of Lines~2--6 is $O(\log n + r)$. The running time of the broadcast in Line~\ref{ln:edge-broadcast} is at most $2 r \beta_b = O(r \mu(H))$. Therefore, the overall running time is $O(\mu(H) r + \log n)$, as claimed.
\end{proof}


\subsubsection{Near Optimality as a Function of \texorpdfstring{$\mu$}{mu}}

Here we show that the $\mu$-dependence in the upper bound of Theorem~\ref{thm:maximum-density-enumeration} is essentially tight for the PC model.

\begin{theorem}\label{thm:maximum-density-lb}
	For every $\eps > 0$, any algorithm that solves triangle or simplex enumeration with sufficiently high probability in hypergraphs with maximum density $\mu$ requires $\Omega(\mu^{1-\eps}/\log n)$ rounds in the PC model.
\end{theorem}
\begin{proof}
	We first observe that a rank-$r$ clique on $n$ vertices has $\mu = \Theta(n^{r-1})$. Thus, for any $\mu$ and $r$, a hypergraph $H$ may contain a clique of size $n' = \mu^{1/(r-1)}$ and still have maximum density $O(\mu)$. By including each edge in such a clique independently with probability $1/2$, we obtain the distribution from the lower bound construction of Theorem~\ref{thm:clique-lb}, for which the implied lower bound is $\Omega((n')^{r - O(1)} / \log n) = \Omega(\mu^{1 - O(1/r)}/\log n)$. Using an identical construction to the proof of Theorem~\ref{thm:sparse-lb}, this bound holds for the PC model for hypergraphs with any number vertices as well. Taking $r$ sufficiently large, we obtain the desired result.
\end{proof}

Again, the lower bound of Theorem~\ref{thm:maximum-density-lb} does not apply to the EC, EB, EU, ES, and EP models because the constructed instances have (expected) maximum 2-degree $\Delta_2 = \Theta(\mu^{(r-2)/(r-1)})$. Thus, applying Proposition~\ref{prop:pc-relationship}, we obtain a lower bound of $\Omega(\mu^{1/3(r-1)})$ for these models. Taking $r = 2$, we obtain the following corollary.

\begin{corollary}\label{cor:maximum-density-lb}
	Any algorithm that solves triangle or simplex enumeration with sufficiently high probability in hypergraphs with maximum density $\mu$ requires $\Omega(\mu^{1/3}/\log n)$ rounds in the EC model.
\end{corollary}








%% file: conclusion.tex
\section{Discussion and Future Work}
\label{sec:conclusion}

In this paper, we have initiated the systematic study of distributed sub-hypergraph enumeration in hypergraphs. To this end, we (1)~introduced several computational models for hypergraphs that generalize the CONGEST model for graphs and evaluate their relative computational power, (2)~provided algorithms for distributed triangle and simplex enumeration in our computational models and proved their optimality in two such models, (3)~introduced classes of sparse and ``everywhere sparse'' hypergraphs and described efficient distributed algorithms for triangle and simplex enumeration in these classes.

While we proved that our triangle enumeration algorithms are optimal or near-optimal for the CLIQUE and PRIMAL CONGEST (PC) models, we do not prove optimality for the EDGE CLIQUE (EC) or weaker models. The reason is that these models provide potentially more computational power by communicating across parallel edges in a hypergraph. For example, consider the distribution of hypergraphs $G(n, r, 1/2)$ used to prove the $\Omega(n^{r - 5/3}/\log n)$ lower bound for triangle enumeration in the CLIQUE and PC models. Observe that if $H = (V, E) \sim G(n, r, 1/2)$, then for every pair of vertices $u, v \in V$, there are $\frac 1 2 \binom{n-2}{r-2}$ edges $e$ containing both $u$ and $v$ in expectation. A simple application of Chernoff bounds shows that for constant $r$, every pair of vertices $u$, $v$ is contained in $\Omega(n^{r-2})$ edges with high probability. Thus, in the EC model, this pair can exchange $\Omega(n^{r-2} \log n)$ bits each round. As a result, the messages from our adaptation of Dolev et al.'s algorithm~\cite{Dolev2012Tri} for triangle enumeration in the CLIQUE model can be simulated using only $O(n^{1/3}/\log n)$ communication rounds (with high probability over selection of $H \sim G(n, r, 1/2)$). We suspect that, at least for constant rank hypergraphs, this upper bound is achievable for all hypergraphs in the EC model.

\begin{question}\label{ques:ec-ub}
	It is possible to perform triangle enumeration in $\widetilde{O}(n^{1/3})$ rounds of the EC model in hypergraphs of bounded rank. Similarly, can we perform simplex enumeration in $\widetilde{O}(n^{1/(r+1)})$ rounds of the EC model?
\end{question}

We suspect an affirmative answer to Question~\ref{ques:ec-ub}. Towards proving such an upper bound, it seems possible that the techniques of Chang et al.~\cite{Chang2021Nearoptimal} could still be applied, albeit with more significant modifications than our discussion in Section~\ref{sec:general-hypergraphs}. Indeed, the main ingredients of Chang et al.---triangle enumeration in rapidly mixing (hyper)graphs and expander decompositions---would both need to be generalized in a nontrivial way. Yet some parts of the machinery employed by Chang et al.\ do generalize in a natural way to the EC model in such a way that parallel edges can be exploited for more efficient message routing. We conclude the paper by describing such techniques that we believe will be useful in addressing triangle enumeration and other problems in hypergraphs. We focus on making algorithms efficient in the EDGE UNICAST (EU) model, which is more restrictive (and likely more realistic) than the EC model in hypergraphs with unbounded rank.

One basic idea for generalizing CONGEST algorithms to the EU model is to simulate the execution of a CONGEST algorithm on a graph embedded in a hypergraph network. One way to an embedded graph in a hypergraph $H = (V, E)$ is by forming a graph $G = (V, E')$ where each edge $e \in E$ is replaced by graph with vertex set $e$. We call such an embedded graph a \dft{graph reduction} of $H$.

\begin{definition}\label{dfn:graph-reduction}
	Let $H = (V, E)$ be a hypergraph and let $\mathcal{G} = \set{G_2, G_3, \ldots,}$ be a family of graphs containing a single graph $G_k$ on $k$ vertices for each integer $k \geq 2$. Then a \dft{graph reduction} of $H$ by $\mathcal{G}$ is a multigraph obtained by replacing each edge $e \in E$ with an isomorphic copy of $G_{\abs{e}}$, where parallel edges may occur between pairs of vertices contained in multiple edges in $E$.
\end{definition}

For suitably chosen families $\mathcal{G}$, graph reductions of a hypergraph can be advantageous for simulating CONGEST algorithms in the EC model. We make the following observations. Throughout, we let $H = (V, E)$ be a hypergraph, $\mathcal{G}$ a family of graphs as above, and $G_H$ the graph reduction of $H$ by $\mathcal{G}$.

\begin{enumerate}
	\item Suppose $\mathcal{G}$ has maximum degree at most $\delta$---i.e., every $G_k \in \mathcal{G}$ has maximum degree at most $\delta$. Then each vertex $v$ in $G_H$ has degree at most $\delta \deg(v)$, where $\deg(v)$ is $v$'s degree in $H$.
	\item If $\mathcal{G}$ has maximum degree at most $\delta$, then a single round of CONGEST in $G_H$ can be simulated in $\delta$ rounds of EU in $H$.
	\item Suppose each graph $G_k \in \mathcal{G}$ has diameter $d(k)$. Then $G_H$ has diameter at most $d(r) D$, where $D$ is the diameter of $H$ and $r$ is its rank.
\end{enumerate}

These observations suggest that many communication primitives such as broadcast and convergecast~\cite{Peleg2000Distributed} can be simulated efficiently in the EU model by simulating CONGEST algorithms in graph reductions. All that is required is to choose graph families $\mathcal{G}$ with bounded degree and small diameter (e.g., $\mathrm{diam}(G_k) = O(\log k)$).

For efficient message routing in hypergraphs, we can employ the framework of Ghaffari et al.~\cite{Ghaffari2017Distributed,Ghaffari2018New}, which solves point-to-point message routing parameterized by the network's mixing time (a crucial ingredient to Chang et al.'s triangle enumeration algorithm~\cite{Chang2021Nearoptimal}). For any hypergraph $H$, we can define a random walk on $H$ according to the following process: to step from a vertex $v$, choose an edge $e$ incident to $v$ uniformly at random, then choose a vertex $u \in e$ uniformly at random. That is, the transition probability to move from a vertex $u$ from $v$ is
\begin{equation}\label{eqn:transition}
	P(v, u) = \frac{1}{\deg(v)} \sum_{e \ni u, v} \frac{1}{\abs{e}}.
\end{equation}
In the case of graphs, these transition probabilities give precisely the lazy random walk on the graph. A straightforward calculation shows that in general, this hypergraph random walk has a stationary distribution that is proportional to the degree distribution.

Using transition probabilities~(\ref{eqn:transition}), we can define the conductance $\Phi$ of the hypergraph $H$ to be the conductance of the underlying Markov chain. As shown by Sinclair and Jerrum~\cite{Sinclair1989Approximate}, the conductance $\Phi$ controls the mixing time $\tmix$ of the random walk $H$. Applying this machinery to a graph reduction $G_H$, one can show that if $H$ has conductance $\Phi$ and the graph family $\mathcal{G}$ is a family of  $k$-regular expander graphs with expansion $\eps$, then the graph reduction $G_H$ has conductance at least $\frac{\eps}{k} \Phi$. Thus, if $H$ is rapidly mixing, then so is the graph reduction $G_H$. We can then simulate the message routing procedures of Ghaffari et al.~\cite{Ghaffari2017Distributed,Ghaffari2018New} in $G_H$ to efficiently route messages in the EC model in $H$.

Applying this approach for message routing, we could then try to generalize the triangle enumeration procedure of Chang et al.~\cite{Chang2021Nearoptimal} to perform triangle enumeration in $\widetilde{O}(n^{1/3})$ rounds of EC. Indeed, a natural generalization of the partition of (possible) triangles used by Chang et al.\ ensures that each vertex $v$ is the source of $O(n^{1/3} \deg(v))$ messages. But in order to route all messages in $\widetilde O(n^{1/3} \tmix)$ rounds, each vertex must also be the destination of $O(\deg(v) n^{1/3})$, which does not seem to be guaranteed by Chang et al.'s partition for hypergraphs. Finding a suitable modification of this method seems to be a promising approach to an improved upper bound for triangle enumeration in rapidly mixing hypergraphs in the EC model. Generalizing expander decompositions to hypergraphs could then yield efficient algorithms for arbitrary hypergraphs as well.

\subsection{Further Directions}

Beyond triangle and simplex enumeration, we hope that the definitions and techniques presented in this paper prove valuable in solving other distributed problems in hypergraphs. One natural extension of triangle or simplex enumeration is to consider more general sub-hypergraph enumeration, such as large cliques or other motifs.


